\documentclass[prb,twocolumn,aps,superscriptaddress,10pt]{revtex4-2}
\usepackage{relsize}
\usepackage{graphicx}
\usepackage{dcolumn}
\usepackage{bm}

\usepackage{natbib}

\usepackage{physics}
\usepackage{amsfonts}
\usepackage{amssymb}
\usepackage{amsthm}
\usepackage{amsmath}
\usepackage{enumitem}

\usepackage{hyperref}
\usepackage[dvipsnames]{xcolor}
\definecolor{col}{RGB}{250, 64, 47}
\hypersetup{colorlinks=true, linkcolor=col, citecolor=col, urlcolor=col}
\usepackage{cleveref}
\newtheorem{theorem}{Theorem}
\newtheorem{result}{Result}

\newtheorem{lem}{Lemma}
\DeclareMathOperator{\poly}{poly}
\renewcommand{\var}{\operatorname{var}}

\newtheorem{defn}{Definition}
\usepackage{mathtools}
\newcommand{\llangle}{\langle\!\langle}
\newcommand{\rrangle}{\rangle\!\rangle}

\usepackage{bbm}
\usepackage{relsize}

\newcommand{\C}{\mathbb{C}}

\DeclarePairedDelimiterX{\inner}[2]{\langle}{\rangle}{#1|#2}
\DeclarePairedDelimiterX{\expect}[3]{\langle}{\rangle}{#1|#2|#3}

\usepackage{tikz}
\usetikzlibrary{decorations.pathreplacing,calligraphy,decorations.markings}
\usetikzlibrary {arrows.meta}

\definecolor{tensorcolor}{rgb}{0.65,0.77,0.95}
\definecolor{btensorcolor}{rgb}{0.65,0.50,0.69}
\definecolor{whitetensorcolor}{rgb}{0.93,0.93,0.93}
\definecolor{gtensorcolor}{rgb}{0.6,0.8,0.5}
\definecolor{operatorcolor}{rgb}{1.0,1.0,1.0}

\newcommand\singledx{1.8}

\newcommand{\GTensor}[5]{
	\begin{scope}[shift={(#1)}]
    \ifnum#5=0
		\draw[very thick, draw=red] (-#2,0) -- (#2,0);
            \draw[very thick] (0,#2) -- (0,0);
    \fi
    \ifnum#5=-1
		\draw[very thick] (0,0) -- (#2,0);
            \draw[very thick] (0,#2) -- (0,0);
    \fi
    \ifnum#5=1
		\draw[very thick] (-#2,0) -- (0,0);
            \draw[very thick] (0,#2) -- (0,0);
    \fi

    \ifnum#5=2
		\draw[very thick,draw=red] (-#2,0) -- (#2,0);
    \fi
    \ifnum#5=3
		\draw[very thick] (0,-#2) -- (0,#2);
    \fi
    \ifnum#5=4
		\draw[very thick] (-#2,0) -- (#2,0);
    \fi
    \ifnum#5=5
		\draw[very thick, draw=red] (-#2,0) -- (#2,0);
		\draw[very thick] (0,#2) -- (0,-#2);
    \fi
    \ifnum#5=6
		\draw[very thick] (-#2,0) -- (#2,0);
            \draw[very thick] (0,#2) -- (0,0);
    \fi
    \ifnum#5=7
		\draw[very thick, draw=red] (-#2,0) -- (#2,0);
            \draw[very thick] (0,-#2) -- (0,0);
    \fi
    \ifnum#5=8
		\draw[very thick] (-#2,0) -- (#2,0);
    \fi
    \ifnum#5=9
		\draw[very thick] (-#2,0) -- (#2,0);
            \draw[very thick] (0,-#2) -- (0,0);
    \fi
        \draw[ thick, fill=tensorcolor, rounded corners=2pt] (-#3,-#3) rectangle (#3,#3);
		\draw (0,0) node {\scriptsize #4};
	\end{scope}
}

\newcommand{\GFTensor}[5]{
	\begin{scope}[shift={(#1)}]
    \ifnum#5=0
		\draw[very thick, draw=red] (-#2,0) -- (\singledx+#2,0);
            \draw[very thick, draw=black] (0,#2) -- (0,0);
            \draw[very thick, draw=black] (\singledx,#2) -- (\singledx,0);
    \fi
    \ifnum#5=1
		\draw[very thick, draw=red] (-#2,0) -- (\singledx+#2,0);
            \draw[very thick, draw=black] (0,-#2) -- (0,0);
            \draw[very thick, draw=black] (\singledx,-#2) -- (\singledx,0);
    \fi
    \ifnum#5=2
		\draw[very thick, draw=red] (0,-#2) -- (0,#2);
            \draw[very thick, draw=red] (\singledx,-#2) -- (\singledx,#2);
    \fi
    \ifnum#5=3
		\draw[very thick] (0,-#2) -- (0,#2);
    \fi
    \ifnum#5=4
		\draw[very thick] (-#2,0) -- (#2,0);
    \fi
        \draw[ thick, fill=tensorcolor, rounded corners=2pt] (-#3,-#3) rectangle (#3+\singledx,#3);
		\draw (0+0.5*\singledx,0) node {\scriptsize #4};
	\end{scope}
}

\newcommand{\UFTensor}[5]{
	\begin{scope}[shift={(#1)}]
    \ifnum#5=0
            \draw[very thick] (0,#2) -- (0,-#2);
            \draw[very thick] (\singledx,#2) -- (\singledx,-#2);
    \fi
    \ifnum#5=1
		\draw[very thick, draw=red] (0,-#2) -- (0,0);
            \draw[very thick, draw=red] (\singledx,-#2) -- (\singledx,0);
            \draw[very thick] (0.5*\singledx,#2) -- (0.5*\singledx,0);
    \fi

    \ifnum#5=2
		\draw[very thick, draw=red] (0,#2) -- (0,0);
            \draw[very thick, draw=red] (\singledx,#2) -- (\singledx,0);
            \draw[very thick] (0.5*\singledx,-#2) -- (0.5*\singledx,0);
    \fi
    \ifnum#5=3
		\draw[very thick, draw=red] (0,#2) -- (0,-#2);
            \draw[very thick, draw=red] (\singledx,#2) -- (\singledx,-#2);
    \fi
    \ifnum#5=4
		\draw[very thick] (-#2,0) -- (#2,0);
    \fi
        \draw[ thick, fill=whitetensorcolor, rounded corners=2pt] (-#3,-#3) rectangle (#3+\singledx,#3);
		\draw (0+0.5*\singledx,0) node {\scriptsize #4};
	\end{scope}
}

\newcommand{\GPTensor}[5]{
	\begin{scope}[shift={(#1)}]
    \ifnum#5=0
		\draw[very thick, draw=red] (-#2,0) -- (#2,0);
            \draw[very thick, draw=red] (0,-#2) -- (0,#2);
    \fi
    \ifnum#5=1
		\draw[very thick, draw=red] (-#2,0.2) -- (#2,0.2);
            \draw[very thick, draw=red] (-#2,-0.2) -- (#2,-0.2);
            \draw[very thick, draw=red] (0.2,-#2) -- (0.2,#2);
            \draw[very thick, draw=red] (-0.2,-#2) -- (-0.2,#2);
    \fi

    \ifnum#5=2
		\draw[very thick, draw=red] (-#2,0) -- (#2,0);
            \draw[very thick, draw=red] (0,-#2) -- (0,#2);
    \fi

    \ifnum#5=3
		\draw[very thick] (0,-#2) -- (0,#2);
    \fi
        \draw[ thick, fill=tensorcolor, rounded corners=2pt] (-#3,-#3) rectangle (#3,#3);
	\draw (0,0) node {\scriptsize #4};
    \ifnum#5=0
		\draw[very thick] (#3/2,#3/2) -- (#2,#2);
    \fi
	\end{scope}
}

\newcommand{\UPTensor}[5]{
	\begin{scope}[shift={(#1)}]
    \ifnum#5=0
		\draw[very thick, draw=red] (-#2,0) -- (#2,0);
            \draw[very thick, draw=red] (0,-#2) -- (0,#2);
    \fi
    \ifnum#5=1
		\draw[very thick, draw=red] (-#2,0.2) -- (#2,0.2);
            \draw[very thick, draw=red] (-#2,-0.2) -- (#2,-0.2);
            \draw[very thick, draw=red] (0.2,-#2) -- (0.2,#2);
            \draw[very thick, draw=red] (-0.2,-#2) -- (-0.2,#2);
    \fi

    \ifnum#5=2
		\draw[very thick] (-#2,0) -- (#2,0);
            \draw[very thick] (0,-#2) -- (0,#2);
    \fi

    \ifnum#5=3
		\draw[very thick, draw=red] (-#2,0) -- (#2,0);
            \draw[very thick, draw=red] (0,-#2) -- (0,#2);
    \fi
        \draw[ thick, fill=whitetensorcolor, rounded corners=2pt] (-#3,-#3) rectangle (#3,#3);
	\draw (0,0) node {\scriptsize #4};
    \ifnum#5=0
		\draw[very thick] (#3/2,#3/2) -- (#2,#2);
    \fi
    \ifnum#5=2
		\draw[very thick] (#3/2,#3/2) -- (0.75*#2,0.75*#2);
    \fi
    \ifnum#5=3
		\draw[very thick] (#3/2,#3/2) -- (0.75*#2,0.75*#2);
    \fi
	\end{scope}
}

\newcommand{\EPTensor}[5]{
	\begin{scope}[shift={(#1)}]
    \ifnum#5=0
		\draw[very thick, draw=red] (-#2,0.2) -- (#2,0.2);
            \draw[very thick, draw=red] (-#2,-0.2) -- (#2,-0.2);
            \draw[very thick, draw=red] (0.2,-#2) -- (0.2,#2);
            \draw[very thick, draw=red] (-0.2,-#2) -- (-0.2,#2);
    \fi
    \ifnum#5=1
		\draw[very thick, draw=red] (-#2,0.2) -- (#2,0.2);
            \draw[very thick, draw=red] (-#2,-0.2) -- (#2,-0.2);
            \draw[very thick, draw=red] (0.2,-#2) -- (0.2,#2);
            \draw[very thick, draw=red] (-0.2,-#2) -- (-0.2,#2);
    \fi

    \ifnum#5=2
		\draw[very thick,draw=red] (-#2,0) -- (#2,0);
    \fi

    \ifnum#5=3
		\draw[very thick] (0,-#2) -- (0,#2);
    \fi
        \draw[ thick, fill=tensorcolor, rounded corners=2pt] (-#3,-#3) rectangle (#3,#3);
	\draw (0,0) node {\scriptsize #4};
	\end{scope}
}

\newcommand{\Unitary}[5]{
	\begin{scope}[shift={(#1)}]
    \ifnum#5=0
		\draw[very thick, draw=red] (-#2,0) -- (#2,0);
            \draw[very thick] (-#2,2*#3) -- (#2,2*#3);
    \fi
    \ifnum#5=-1
		\draw[very thick] (0,0) -- (#2,0);
            \draw[very thick] (0,#2) -- (0,0);
    \fi
    \ifnum#5=1
		\draw[very thick, draw=red] (-#2,0) -- (#2,0);
            \draw[very thick] (0,2*#3) -- (#2,2*#3);
    \fi

    \ifnum#5=2
		\draw[very thick, draw=red] (-#2,0) -- (#2,0);
            \draw[very thick, draw=red] (-#2,2*#3) -- (#2,2*#3);
    \fi

    \ifnum#5=3
		\draw[very thick] (0,-#2) -- (0,#2);
    \fi
        \draw[ thick, fill=tensorcolor, rounded corners=2pt] (-#3,-#3) rectangle (#3,3*#3);
		\draw (0,#3) node {\scriptsize #4};
	\end{scope}
}

\newcommand{\FUnitary}[5]{
	\begin{scope}[shift={(#1)}]
    \ifnum#5=0
		\draw[very thick, draw=red] (-#2,0) -- (#2,0);
            \draw[very thick, draw=red] (-#2,2*#3) -- (#2,2*#3);
            \draw[very thick, draw=red] (-#2,4*#3) -- (#2,4*#3);
    \fi
    \ifnum#5=1
		\draw[very thick, draw=red] (-#2,0) -- (#2,0);
            \draw[very thick, draw=red] (0,2*#3) -- (#2,2*#3);
            \draw[very thick, draw=red] (0,4*#3) -- (#2,4*#3);
    \fi
    \ifnum#5=2
		\draw[very thick,draw=red] (-#2,0) -- (#2,0);
    \fi

    \ifnum#5=3
		\draw[very thick] (0,-#2) -- (0,#2);
    \fi
        \draw[ thick, fill=tensorcolor, rounded corners=2pt] (-#3,-#3) rectangle (#3,5*#3);
		\draw (0,2*#3) node {\scriptsize #4};
	\end{scope}
}

\newcommand{\PTensor}[5]{
	\begin{scope}[shift={(#1)}]
    \ifnum#5=0
		\draw[very thick, draw = red] (0,-#2) -- (0,0);
            \draw[very thick, draw = red] (0,0) -- (0,#2);
    \fi
    \ifnum#5=1
		\draw[very thick, draw = red] (-#2,0) -- (0,0);
            \draw[very thick, draw = red] (0,0) -- (#2,0);
    \fi
    \ifnum#5=2
		\draw[very thick, draw = red] (0,-#2) -- (0,0);
            \draw[very thick, draw = red] (0,0) -- (0,#2);
    \fi
    \ifnum#5=3
		\draw[very thick, draw = red] (-#2,0) -- (0,0);
            \draw[very thick, draw = red] (0,0) -- (#2,0);
    \fi
        \draw[ thick, fill=tensorcolor, rounded corners=2pt] (-#3,-#3) rectangle (#3,#3);
		\draw (0,0) node {\scriptsize #4};
	\end{scope}
}

\newcommand{\BTensor}[5]{
	\begin{scope}[shift={(#1)}]
    \ifnum#5=0
		\draw[very thick, draw=red] (-#2,0) -- (#2,0);
    \fi
    \ifnum#5=1
		\draw[very thick,draw=red] (0,#2) -- (0,-#2);
    \fi
    \ifnum#5=2
		\draw[very thick] (-#2,0) -- (#2,0);
    \fi
    \ifnum#5=3
		\draw[very thick] (0,#2) -- (0,-#2);
    \fi

        \draw[ thick, fill=tensorcolor, rounded corners=2pt] (-#3,-#3) rectangle (#3,#3);
		\draw (0,0) node {\scriptsize #4};
	\end{scope}
}

\newcommand{\DTensor}[5]{
	\begin{scope}[shift={(#1)}]
    \ifnum#5=0
		\draw[very thick, draw=red] (-#2,0) -- (#2,0);
    \fi
    \ifnum#5=1
		\draw[very thick,draw=red] (0,#2) -- (0,-#2);
    \fi
    \ifnum#5=2
		\draw[very thick] (-#2,0) -- (#2,0);
    \fi
    \ifnum#5=3
		\draw[very thick] (0,#2) -- (0,-#2);
    \fi

        \draw[ thick, fill=whitetensorcolor, rounded corners=2pt] (-#3,-#3) rectangle (#3,#3);
		\draw (0,0) node {\scriptsize #4};
	\end{scope}
}

\newcommand{\UTensor}[5]{
	\begin{scope}[shift={(#1)}]
    \ifnum#5=0
		\draw[very thick, draw=red] (-#2,0) -- (#2,0);
    \fi
    \ifnum#5=1
		\draw[very thick,draw=red] (0,#2) -- (0,-#2);
    \fi
    \ifnum#5=2
		\draw[very thick] (-#2,0) -- (#2,0);
    \fi
    \ifnum#5=3
		\draw[very thick] (0,#2) -- (0,-#2);
    \fi
        \draw[ thick, fill=operatorcolor, rounded corners=2pt] (-#3,-#3) rectangle (#3,#3);
		\draw (0,0) node {\scriptsize #4};
	\end{scope}
}

\newcommand{\RTensor}[5]{
	\begin{scope}[shift={(#1)}]
    \ifnum#5=0
		\draw[very thick, draw=red] (-#2,0) -- (#2,0);
    \fi
    \ifnum#5=1
		\draw[very thick,draw=black] (0,#2) -- (0,-#2);
    \fi
    \ifnum#5=2
		\draw[very thick] (-#2,0) -- (#2,0);
    \fi
    \ifnum#5=3
		\draw[very thick,draw=red] (-#2,0) -- (#2,0);
    \fi
        \draw[ thick, fill=white, rounded corners=2pt] (-#3,-#3) rectangle (#3,#3);
		\draw (0,0) node {\scriptsize #4};
	\end{scope}
}

\newcommand{\GDTensor}[5]{
	\begin{scope}[shift={(#1)}]
    \ifnum#5=0
		\draw[very thick] (-#2,0) -- (#2,0);
		\draw[very thick] (0,#2) -- (0,-#2);
    \fi
    \ifnum#5=-1
		\draw[very thick] (0,0) -- (#2,0);
		\draw[very thick] (0,#2) -- (0,-#2);
    \fi
    \ifnum#5=1
		\draw[very thick] (-#2,0) -- (0,0);
		\draw[very thick] (0,#2) -- (0,-#2);
    \fi
        \draw[ thick, fill=tensorcolor, rounded corners=2pt] (-#3,-#3) rectangle (#3,#3);
    \def\dx{#3/3};
	\draw [thick]  (-#3+\dx, \dx) -- (- \dx,#3-\dx);
	\draw [thick] (-#3+1.5*\dx,-#3+1.5*\dx) -- (#3-1.5*\dx,#3-1.5*\dx);
	\draw [thick]  ( \dx, -#3 + \dx) -- (#3 - \dx,-\dx);
	\draw (0,0) node {\scriptsize #4};
	\end{scope}
}

\newcommand\subsetsim{\mathrel{%
  \ooalign{\raise0.2ex\hbox{$\subset$}\cr\hidewidth\raise-0.8ex\hbox{\scalebox{0.9}{$\sim$}}\hidewidth\cr}}}

\begin{document}


\title{Metrology of open quantum systems from emitted radiation}


\author{Siddhant Midha}
\email{siddhantm@princeton.edu}
 \affiliation{Princeton Quantum Initiative, Princeton University, Princeton, New Jersey 08540, USA}

\author{Sarang Gopalakrishnan}
\email{sgopalakrishnan@princeton.edu}
 \affiliation{Princeton Quantum Initiative, Princeton University, Princeton, New Jersey 08540, USA}
\affiliation{Department of Electrical and Computer Engineering, Princeton University, Princeton, New Jersey 08540, USA}
\date{\today}

\begin{abstract}
We explore the task of learning about the dynamics of a Markovian open quantum system by monitoring the information it radiates into its environment. For an open system with Hilbert space dimension $D$, the quantum state of the emitted radiation can be described as a temporally ordered matrix-product state (MPS). 
%
We provide simple analytical expressions for the quantum Fisher information (QFI) of the radiation state, which asymptotically scales linearly with the sensing time unless the open system has multiple steady states. We characterize the crossovers in QFI near dynamical phase transitions, emphasizing the role of temporal correlations in setting the asymptotic rate at which QFI increases. We discuss when optimal sensing is possible with instantaneously measured radiation.

\end{abstract}

\maketitle

\textit{Introduction.---}In the standard setting for quantum metrology, a known initial state is evolved under dynamics depending on a parameter $\theta$ (e.g., a field) that one is trying to learn~\cite{QFI_Multiparameter_Review,Caves_QFI,Fisher_NISQ_Meyer,MB_metrology_review,giovannetti_metrology}. After the state has evolved for some time $t$, one measures it in some basis, and repeats the experiment until the parameter has been estimated to the desired precision. The precision one can achieve with $m$ tries on a system of $N$ qubits scales as $1/(N \sqrt{m})$, the ``Heisenberg limit,'' which is achieved for certain entangled initial states. The dynamics one is trying to sense could be either Hamiltonian or dissipative; in the standard setting, dissipation hinders learning, and makes the Heisenberg limit unattainable \cite{normbounds2_demkowicz2012elusive,normbounds3_demkowicz2014using,normbounds5_zhou2018achieving,zhou2024limits}. However, dissipation also allows new protocols for sensing: instead of measuring the system, one can monitor the information it leaks into its environment, and use this measurement record to infer $\theta$. Sensing through this indirect protocol has yet lacked a unifying description (despite recent progress~\cite{Lesanovsky_BTC,theodoros_criticality,molmer_14,plenio_decoder,Catana_continuous,Gammelmark_continuous,Kiilerich_photocounting,Kiilerich_homodyne,landi_24,FI_stoch_landi,FI_stoch_binder,FI_jump_binder,Garrahan_16,DPT_sensing1,DPT_sensing2,DPT_sensing3,DPT_sensing4,albarelli2018restoring,albarelli2017ultimate,noisyQND}); however, the closely related problem of learning an initial \emph{state} (as opposed to a dynamical parameter) from the measurement record has seen extensive recent work in the context of measurement-induced phase transitions (MIPTs)~\cite{PhysRevX.9.031009, PhysRevB.100.134306,Potter_2022,annurev:/content/journals/10.1146/annurev-conmatphys-031720-030658, dehghani_neural-network_2023, PRXQuantum.5.030311, PhysRevLett.129.200602, PhysRevLett.130.220404, PhysRevLett.128.050602}. 

In the present work we explore this ``radiation-sensing'' protocol for general finite-sized quantum systems coupled to Markovian environments. Our techniques apply equally to discrete and continuous time, but we will focus on discrete-time processes generated by the iterated application of quantum channels: the radiation collected from a discrete-time process over $t$ time-steps forms a matrix-product state (MPS)~\cite{seq_gen_MPS_1,seq_gen_MPS_2,seq_gen_Astrakhantsev,seq_gen_PEPS,seq_gen_Astrakhantsev,seq_gen_holo1,seq_gen_holo2,seq_gen_holo3,PDA_sarang,seq_gen_nisq} on $t$ qudits [see Fig.~\ref{fig:fig1}], and the quantum channel is the transfer matrix of this MPS. This correspondence lets us use MPS methods to derive simple expressions for the sensing power of the emitted radiation, in terms of the quantum Fisher information (QFI) of the MPS. The QFI lower-bounds the variance that can be achieved using arbitrary measurements on the MPS. For simplicity, we will only consider time-periodic quantum processes, so that the MPS is translation-invariant. However, the generalization to time-dependent processes would be a straightforward extension of our formalism. 

\begin{figure}[b]  
    \centering
    \includegraphics[width=1.0\linewidth]{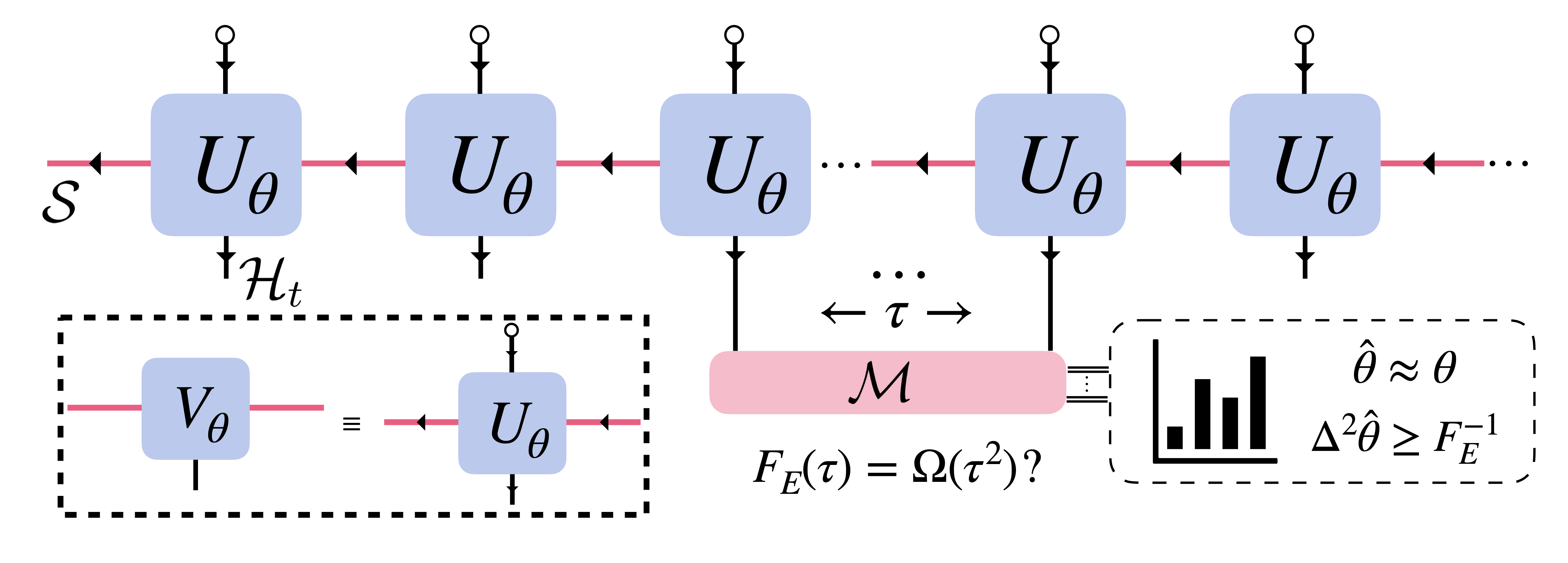}  
    \caption{Quantum sensing of monitored systems. The parameterized unitary mediating the emitter-radiation interaction $U_\theta: \mathcal{S}\otimes\mathcal{H}_t\to \mathcal{S}\otimes\mathcal{H}_t$ induces an isometry $V_\theta \equiv U_\theta \ket{0}_d$. The task of estimating the parameter is rooted in the scaling of the quantum Fisher information $F(\tau)$ with the number of emitted radiation quanta $\tau$, with the ultimate precision determined by the Cram\'{e}r-Rao bound.}
    \label{fig:fig1}  
\end{figure}

Our main results are as follows. First, we derive a simple expression for the QFI of a general quantum channel with a spectral gap $\Delta$~\footnote{For a quantum channel $\mathcal{E}_t$ generated by $\mathcal{L}$, i.e., $\mathcal{E}_t = \exp{\mathcal{L}t}$ we define the spectral gap $\Delta$ as the magnitude of the second largest eigenvalue of $\mathcal{L}$}; this corresponds to an normal MPS, and is generic for finite-dimensional open systems. Our expression immediately implies that for \emph{any} parameter estimation task, the maximum precision achievable as $t \to \infty$ scales as $(\log{D})^{-1}\sqrt{\Delta/(\Delta+1)t}$. For \emph{any} process on a finite-size open system, the achievable precision is limited to $\Omega(1/t)$ \footnote{We use standard asymptotic notation: we say that $f = O(g)$ if there exists a $c > 0$ and $x_0$ such that $f(x) \leq c g(x)$ for $x \geq x_0$. We denote $f=\Omega(g)$ if $g = O(f)$, and $f = \Theta(g)$ if $f = O(g)$ and $g =O(f)$.}. Second, we illustrate this formalism for the concrete case of a ``dissipative time crystal'' on $N$ qubits~\cite{Iemini_BTC,Lesanovsky_BTC,Metrology_2_BTC,Metrology_BTC}. In the time crystal, the long-time and large-size limits do not commute, since $\Delta \to 0$ as $N \to \infty$: at times $t \ll N$, the QFI scales as $N^2 t^2$ (so the achievable precision scales optimally as $1/(N t)$), while for $t \gg N$ the QFI scales as $N^3 t$. Finally, we explore the case of channels with multiple precisely degenerate steady states (e.g., protected by a symmetry). We write down general conditions for the QFI to scale as $t^2$ for arbitrarily late times. We consider a concrete class of such channels, for which the radiation state is a GHZ-type state. We identify conditions under which single-qudit measurements on the emitted radiation saturate the optimal scaling.

\textit{Preliminaries.---}Metrology addresses the following task: given a family of states, $\rho_\theta$, parameterized by the continuous variable $\theta$, one aims to construct a locally unbiased estimator $\hat{\theta}$ based on arbitrary measurements of $\rho_\theta$. 
The fundamental precision limit is set by the quantum Cram\'{e}r-Rao bound \cite{Caves_QFI}, which dictates that $(\Delta\hat{\theta})^2 \geq 1/F(\rho_{\theta})$, where $F(\rho_{\theta})$ is the \textit{quantum Fisher information} (QFI) of the state $\rho_\theta$ (see Sec.~\ref{sec:QFI} of SM \cite{SM}) and $\Delta\hat{\theta}$ is the estimation error. The QFI is given as $F(\rho_\theta) = \text{Tr}[\rho_\theta L_\theta^2]$, where $L_\theta$ is the symmetric logarithmic derivative implicitly defined as $\partial_\theta\rho_\theta =: \{\rho_\theta,L_\theta\}/2$. For pure states $\ket{\psi_\theta}$, this reduces to $F(\psi_\theta) = 4\left(\inner{\dot{\psi}_\theta}{\dot{\psi}_\theta} - |\inner{\dot{\psi}_\theta}{\psi_\theta}|^2\right)$ with $\ket{\dot{\psi}_\theta} = \partial_\theta\ket{\psi_\theta}$. 

In the setting we consider, the open system (or ``emitter'') is a state in a Hilbert space $\mathcal{S} \cong \mathbb{C}^D$. It will suffice to consider an initial pure state $\ket{\psi_0}$ in $\mathcal{S}$. Because we are considering Markovian dynamics, at each time step $t$, a fresh Hilbert space $\mathcal{H}_t \cong \mathbb{C}^d$ (consisting of a $d$-dimensional qudit) is introduced, and an isometry $V_\theta: \mathcal{S} \to \mathcal{S} \otimes \mathcal{H}_t$ is applied. For simplicity we specialize to the time-translation invariant case where $V_\theta$ does not depend on $t$. 
After the system evolves for $T$ time steps, the full state of the system and universe is a pure state $\ket{\Psi; t} \in \mathcal{S} \otimes \prod_{t = 1}^T \mathcal{H}_t$. To get a Trotterized version of the continuous-time dynamics one can split that evolution into steps of size $dt$ and derive the corresponding isometry for that infinitesimal time step, which will approach the identity as $dt \to 0$. Tracing out $\mathcal{H}_t$ in this limit gives a Lindblad master equation. The general form of the Lindblad master equation is $\partial_t \rho = -i [H, \rho] + \sum_m (L_m \rho L_m^\dagger - \{L_m^\dagger L_m, \rho\}/2) \equiv \mathcal{L}(\rho)$, where $H$ is the Hamiltonian and $L_m$ are jump operators, and $\mathcal{L}$ is the Liouvillian generating the dynamics. The radiation state in this limit is a ``continuous MPS''~\cite{PhysRevLett.104.190405}. We will develop the formalism in the discrete case as it is more transparent.

The state $\ket{\Psi; t}$ defines a hierarchy of QFIs [see also Sec.~\ref{sec:hierarcy} of \cite{SM}], in which the lower levels correspond to marginals of $\ket{\Psi; t}$. Recall that the QFI monotonically decreases under partial traces, so each level is upper-bounded by those above it. At the top of the hierarchy, there is the full state $\ket{\Psi; t}$, which contains all possibly attainable information about $\theta$. In general, we do not have access to $\mathcal{S}$; instead, the ``radiation QFI'' is encoded in $\rho_R \equiv \mathrm{Tr}_{\mathcal{S}} (\ket{\Psi;t} \bra{\Psi;t})$. (Also at this level is the usual noisy QFI, which is encoded in the complementary trace over all the $\mathcal{H}_t$.) One of our results is that to leading order in $T$, the radiation QFI coincides with the full QFI, assuming $D$ is finite. The radiation QFI has a simple expression for an MPS, so we will focus on it; however, it is challenging to measure in experiments. In typical experiments, the emitted radiation consists of photons, which are immediately detected or lost; on the other hand, the radiation QFI assumes access to entangled measurements between radiation emitted at different times. To get a more experimentally accessible quantity, one can apply a dephasing channel to $\rho_R$, getting $\rho_D = \sum_{\mathbf{m}} p_{\mathbf{m}} \ket{\mathbf{m}} \bra{\mathbf{m}}$, where $\{ \mathbf{m} \}$ is a density matrix that is diagonal in some product basis on $\prod_t \mathcal{H}_t$. (Note that this product basis need not be time-independent; for example, it includes single-site shadows~\cite{huang2020predicting}.) While $\rho_D$ is the quantity of greatest interest, however, we are only able to estimate it under specific assumptions about the dynamics.

It is immediate to see that the state $\ket{\Psi;t}$ can be expressed as a translation-invariant MPS (TI-MPS, see Sec.~\ref{sec:TIMPS} of \cite{SM}) of bond dimension $D$, as follows:

\begin{align}
   \ket{\Psi} =  \dotso \quad
\begin{tikzpicture}[scale=0.5,baseline={([yshift=-0.65ex] current bounding box.center) }]
        \GTensor{0,0}{1.2}{.6}{\small $V$}{7};
      \GTensor{1*\singledx,0}{1.2}{.6}{\small $V$}{7};
      \GTensor{2*\singledx,0}{1.2}{.6}{\small $V$}{7};
      \GTensor{3*\singledx,0}{1.2}{.6}{\small $V$}{7};
      \GTensor{4*\singledx,0}{1.2}{.6}{\small $V$}{7};
\end{tikzpicture}
    \quad \dotso
\end{align}
Here, each object $V$ is a $D \times D \times d$ tensor corresponding to the isometry introduced above. To each TI-MPS one can associate the ``transfer matrix'' $T_V$, defined as 
\begin{align}
  T_V := \sum_{m=0}^{d-1}K_m\otimes\bar{K}_m \quad \equiv 
\begin{tikzpicture}[scale=0.5,baseline={([yshift=-0.65ex] current bounding box.center) }]
      \GTensor{0,0}{1.2}{.6}{\small $V$}{7};
      \GTensor{0,-2}{1.2}{.6}{\small $\bar{V}$}{0};
    \end{tikzpicture}
    \quad  
\end{align} 
which is simply the superoperator on $\mathcal{S}$ that generates the open-system dynamics with the channel $\mathcal{E}(\rho) = \sum_m K_m\rho K_m^\dagger$. When the dynamics has a unique steady state, 
the transfer matrix admits the following description,
\begin{align}
  T_V = |\rho_{ss}\rrangle\llangle \mathbbm{1}| + \tilde{T}_V  \equiv \quad
    \begin{tikzpicture}[scale=0.5,baseline={([yshift=-0.65ex] current bounding box.center) }]
        \DTensor{0,-1}{1.}{.6}{\small $\rho_{ss}$}{1};
        \draw[very thick,draw=red] (-1,0)--(0,0);
            \draw[very thick,draw=red] (-1,-2)--(0,-2);
            \draw[very thick,draw=red] (1.2,0)--(2,0);
            \draw[very thick,draw=red] (1.2,-2)--(2,-2);
            \draw[very thick,draw=red] (1.2,-2)--(1.2,0);
    \end{tikzpicture}
        \quad + \tilde{T}_V
\end{align} 
where $\tilde{T}_V$ has spectral radius less than unity, $|\Lambda(\tilde{T}_V)| < 1$, and $\mathcal{E}(\rho_{ss}) = \rho_{ss}$ is a fixed point of the channel. The vectorization of an operator $A$ is denoted $|A\rrangle$. The corresponding MPS is called a normal MPS. For the set of normal MPS, local expectation values are well defined and independent of the
boundary matrix \cite{Haegeman_geometry,Haegeman_postMPSmethods}.

\textit{QFI for normal MPS.---}
Our first main result concerns the full QFI for a gapped quantum channel, or equivalently for a normal MPS, acting for $T$ time steps:

\begin{equation} \label{eq:QFI_TIMPS}
  F(t;V,\theta) = 4\left(T\alpha + 2\sum_{\tau =0}^{T-2}(T-\tau-1)\beta_\tau - T^2|\gamma|^2\right).
\end{equation}
where we define
\begin{align}\label{eq:alphanormal}
  \alpha := \quad
\begin{tikzpicture}[scale=0.5,baseline={([yshift=-0.65ex] current bounding box.center) }]
      \GTensor{0,0}{1.2}{.6}{\small $\dot{V}$}{7};
      \GTensor{0,-2}{1.2}{.6}{\small $\dot{\bar{V}}$}{0};
      \draw[very thick,draw=red] (1.2,0)--(2.03,0);
    \DTensor{2,-1}{1.}{.6}{\small $\rho_{ss}$}{1};
          \draw[very thick,draw=red] (1.2,-2)--(2.03,-2);
        \draw[very thick,draw=red] (-2,0)--(-1.2,0);
        \draw[very thick,draw=red] (-2,-2)--(-1.2,-2);
        \draw[very thick,draw=red] (-2,-2)--(-2,0);
\end{tikzpicture}
    \quad 
\gamma :=  \quad
\begin{tikzpicture}[scale=0.5,baseline={([yshift=-0.65ex] current bounding box.center) }]
      \GTensor{0,0}{1.2}{.6}{\small $\dot{V}$}{7};
      \GTensor{0,-2}{1.2}{.6}{\small $\bar{V}$}{0};
      \draw[very thick,draw=red] (1.2,0)--(2.03,0);
    \DTensor{2,-1}{1.}{.6}{\small $\rho_{ss}$}{1};
          \draw[very thick,draw=red] (1.2,-2)--(2.03,-2);
        \draw[very thick,draw=red] (-2,0)--(-1.2,0);
        \draw[very thick,draw=red] (-2,-2)--(-1.2,-2);
        \draw[very thick,draw=red] (-2,-2)--(-2,0);
\end{tikzpicture} 
\quad
\end{align} 

\begin{align}\label{eq:betasnormal}
  \beta_\tau := \quad
\begin{tikzpicture}[scale=0.5,baseline={([yshift=-0.65ex] current bounding box.center) }]
      \GTensor{-5,0}{1.2}{.6}{\small $\dot{V}$}{7};
      \GTensor{-5,-2}{1.2}{.6}{\small $\bar{V}$}{0};
        \draw[very thick,draw=red] (-5-1.15,-2)--(-5-1.15,0);
        \FUnitary{-3.25,-2.25}{1}{.6}{\small $T_V^\tau$}{5};
      \GTensor{-1.45,0}{1.2}{.6}{\small ${V}$}{7};
      \GTensor{-1.45,-2}{1.2}{.6}{\small $\dot{\bar{V}}$}{0};
          \DTensor{.5,-1}{1.}{.6}{\small $\rho_{ss}$}{1};
          \draw[very thick,draw=red] (-0.5,-2)--(0.53,-2);
          \draw[very thick,draw=red] (-0.5,0)--(0.53,0);
      \end{tikzpicture}
    \quad \forall \tau \geq 0
\end{align}

\noindent as a function of the parametrized isometry $V_\theta$, its derivative $\dot{V}_\theta \equiv \partial_\theta V_\theta$, the transfer matrix $T_V$ along with the steady state $T_V|\rho_{ss}\rrangle = |\rho_{ss}\rrangle$. 
Deriving this equation is a standard exercise in TI-MPS methods, see Sec.~\ref{sec:qfinormalproof} of SM \cite{SM}. 
The channel (by assumption) relaxes any state (in the virtual space) to the steady state on some timescale $\tau^* \sim 1/\Delta$. When $T \gg \tau^*$, one can classify the terms appearing in the QFI into two types: (i)~$O(T)$ ``bulk'' terms, in which the derivatives act at times $\gg \tau^*$ from either the beginning or the end of the process; and (ii)~$O(1)$ ``boundary'' terms, where the derivatives act near the ends of the MPS. Terms of type~(ii) contribute at most $O(\tau^* (\log{D})^k)$ to the QFI, and are therefore subleading as $T \to \infty$. The exponent \(k\) depends on the specific form of parameter encoding—typically \(k = 2\) for single-body Hamiltonian encoding.

For the same reason, the leading-order QFI of the full system coincides with that of the radiation subsystem in the large-\(T\) limit: the extensive (i.e., \(O(T)\)) terms in the QFI are insensitive to the boundary tensor. This can be shown using the purification method of evaluating mixed state QFI [see Lemma \ref{lem:mixedqfibound} in \cite{SM}]. Specifically, we have for large $T$,
\begin{equation}
  F_R(T) \gtrsim F_{SR}(T) - c \tau^*
\end{equation}
for some $c = O((\log{D})^k)$. In fact, if the initial state is Haar random, one can show a similar result for the initial state dependence [see Sec.~\ref{sec:jointQFI} of \cite{SM}]: the effect of an initial state decays asymptotically as $O((\log{D})^ke^{-\Delta t})$. These observations ensure that the asymptotic scaling of the QFI with time is governed solely by the channel dynamics.


We now take the late-time limit of Eq.~\eqref{eq:QFI_TIMPS}. It is helpful to define the quantity $\beta_\infty \equiv \lim_{\tau \to \infty} \beta_\tau$. Uniqueness of the steady state implies that
\begin{align} \label{eq:betarelaxation}
  \beta_\infty = \quad
\begin{tikzpicture}[scale=0.5,baseline={([yshift=-0.65ex] current bounding box.center) }]
      \GTensor{-5,0}{1.2}{.6}{\small $\dot{V}$}{7};
      \GTensor{-5,-2}{1.2}{.6}{\small $\bar{V}$}{0};
        \draw[very thick,draw=red] (-5-1.15,-2)--(-5-1.15,0);
          \draw[very thick,draw=red] (-3.8,-2)--(-3,-2);
          \draw[very thick,draw=red] (-3.8,0)--(-3,0);
          \DTensor{-3.03,-1}{1.}{.6}{\small $\rho_{ss}$}{1};
        \GTensor{-0.5,0}{1.2}{.6}{\small ${V}$}{7};
      \GTensor{-0.5,-2}{1.2}{.6}{\small $\dot{\bar{V}}$}{0};
        \draw[very thick,draw=red] (-.5-1.15,-2)--(-.5-1.15,0);
          \DTensor{1.5,-1}{1.}{.6}{\small $\rho_{ss}$}{1};
          \draw[very thick,draw=red] (0.5,-2)--(1.5,-2);
          \draw[very thick,draw=red] (0.5,0)--(1.5,0);
      \end{tikzpicture}
    \quad  = |\gamma|^2.
\end{align} 
By writing $\beta_\tau = (\beta_\tau - \beta_\infty) + \beta_\infty$, we can express Eq.~\eqref{eq:QFI_TIMPS} as $F = 4T (f_0 + f_c)$, where 
\begin{equation}\label{eq:decomp}
f_0 = (\alpha - |\gamma|^2), \quad 
f_c = 2\sum_{\tau = 0}^{T-2} (\beta_\tau - \beta_\infty)/T.
\end{equation}
Here $f_0$ describes the sensitivity to $\theta$ of radiation measurements at each time (which can be thought of as adding up independently across time), while $f_c$ captures the effects of temporal correlations in the radiation. In general, $f_c$ will depend on the entire spectrum of the channel; when the correlation time $\tau^*$ is long, $f_c \sim \tau^*$ [see Sec.~\ref{subsec:asymprate} of \cite{SM} for further details on the asymptotic rate] and dominates the QFI~\cref{eq:QFI_TIMPS}.

It is instructive to apply this formalism to the general problem of estimating a parameter encoded in the Hamiltonian, $H_\theta \equiv \theta H_s + H_c$, with $H_s$ and $H_c$ denoting the sensing and control Hamiltonians respectively. We assume that the dissipative couplings do not depend on $\theta$. Then $\alpha = \langle H_s^2\rangle_{\rho_{ss}}dt^2$, and $\gamma=-\iota\langle H_s\rangle_{\rho_{ss}}dt + O(dt^2)$. Thus  $f_0 = \var_{\rho_{ss}}[H_s]dt$---i.e., the variance of the sensing Hamiltonian in the steady state---reproducing a standard result in metrology. If the sensing Hamiltonian consists of one-body terms, $\var_{\rho_{ss}}[H_s] = O((\log{D})^2)$, as the square of spectral width upper bounds the variance~\footnote{For any $H_s$ and any state $\rho$, $\var_{\rho}[H_s] \leq (\lambda^s_{\max}-\lambda^s_{\min})^2/4$ where $\lambda^s_{\max,\min}$ are the maximum and minimum eigenvalues of $H_s$. For a $N$-particle single-body Hamiltonian, $\lambda^s_{\max}-\lambda^s_{\min} = O(N)$.}. The term $f_c$ is nonuniversal and depends on the details of the system-bath coupling. However, $f_c = O((\log{D})^2)$ as well for the case above.  

We briefly comment on the results derived with the MPS representation in context to the widely used formula previously as derived in Ref.~\cite{molmer_14}. We find that the asymptotic rate, as derived by a perturbation theory in that formula is identical to the rate obtained from the MPS calculations upon taking the continuous time limit. The MPS calculation, through a trotterization of the continuous approach, furnishes a much simpler analytical expression with evident timescales, and easily extended to the long-range case [see Sec.~\ref{sec:molmerformula} of \cite{SM}].

The continuous time behavior can be recovered with the \cref{eq:QFI_TIMPS} by taking the limit \(dt \to 0\). In the case of Hamiltonian sensing, the contribution from \(f_0\) vanishes in this limit: both \(\alpha\) and \(|\gamma|^2\) scale as \(O(dt^2)\), implying that \(T f_0 \sim dt \to 0\). In contrast, the spectral sum appearing in \(f_c\) exhibits a more subtle dependence on \(dt\), and we find that \(T f_c\) remains finite as \(dt \to 0\). For more general parameter encodings that include jump (dissipative) terms, the \(f_0\) contribution also survives in the continuous-time limit [see Sec.~\ref{subsec:contlimit} of \cite{SM} for details]. We expect the trotterized contribution to be especially relevant in practical scenarios involving time-binned photon monitoring.

\begin{figure}[t]   
    \centering
    \includegraphics[width=1.0\linewidth]{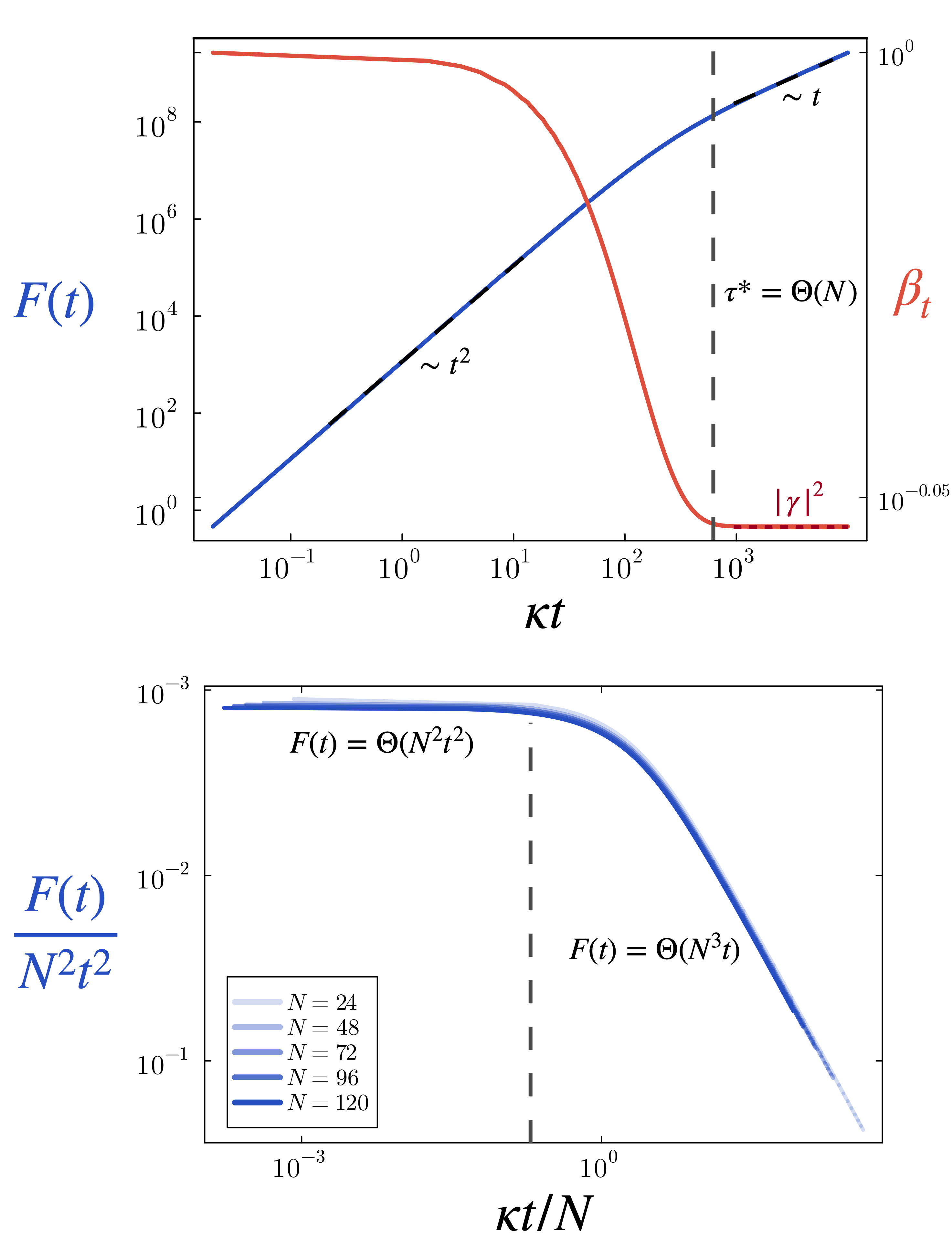}  
    \caption{Continuous sensing with the time crystal. (a) Dynamics of the QFI $F(t)$ computed from \cref{eq:QFI_TIMPS} along with $\beta_t$ for $N=120$ spins at $\omega/\kappa=10$. (b) Finite size scaling of the QFI at $\omega/\kappa = 10$.}
    \label{fig:fig2}  
\end{figure}

\textit{Radiation QFI of the continuous time crystal.---}We apply the formalism discussed to metrology with the continuous time crystal \cite{Iemini_BTC}, an dissipative phase of matter shown to spontaneously break time translation symmetry. This phase has been studied as a metrological resource \cite{Lesanovsky_BTC,Metrology_2_BTC,Metrology_BTC,jumptraj_Lesanovsky}, particularly in the case of continuous sensing \cite{Lesanovsky_BTC}. The model is specified by a Lindblad master equation with collective driving $H=\omega \hat{S}^x$ and  dissipation given by the single jump operator $L=\sqrt{\kappa/S}\hat{S}^-$ where $\hat{S}^{\alpha} = \sum_{i=1}^N \sigma_j^\alpha/2$ and $S=N/2$ is the total spin. This model shows a phase transition at $\omega/\kappa=1$ with the time crystal phase appearing for $\omega/\kappa > 1$. For a finite size system, the gap of the Lindbladian $\Delta$ decreases inversely with the system size $\Delta \sim 1/N$ in the time crystal phase \cite{Iemini_BTC}. Moreover, the steady state is $\rho_{ss} \propto \eta \eta^\dagger$ with $\eta = \sum_{j=0}^N(2\iota\kappa\hat{S}^-/\omega N)^j$ \cite{puri1979exact,GMC_BTC}. The metrological task is to estimate $\omega$.

The characteristic relaxation timescale of the dynamics is set by $\tau^* \sim 1/\Delta \sim N$; accordingly, the QFI has two regimes, $t \ll N$ and $t \gg N$. We begin by considering times $t \gg N$, where we can use the asymptotic expression \cref{eq:decomp}. 
First, we note that the variance of $H$ in the steady state $\rho_{\mathrm{ss}}$ scales as $N^2$, so all individual terms in \cref{eq:decomp} scale as $N^2$. Moreover, the QFI scales linearly with $t$, since the channel has a gap for any finite $N$. 
Thus, the ``static'' part of the QFI, $f_0$, scales as $N^2$. However, since $\beta_\tau - \beta_\infty \sim N^2 \exp(-\tau/N)$, the temporally correlated part of the QFI, $f_c$, in fact scales as $N^3 t$, with an additional factor of $N$ coming from the correlation time. Putting these together, we have $F(t) = (c_1N^2 + c_2 N^3)t$ for some constants with $c_1 = O(dt)$ and $c_2 = O(1)$.  Note that the $N^3$ scaling does not violate a fundamental quadratic bound on QFI, as this is only valid for $T > N$. At times $t \ll N$, instead, the QFI scales as $N^2 t^2$, saturating the Heisenberg limit.

The numerical results for the radiation QFI of $N=120$ spins deep in the time crystal phase at $\omega/\kappa=10$ are shown in Fig.~\ref{fig:fig2}(a). We find two regimes, with quadratic and linear QFI respectively, separated by a transition time $\tau^*$. As expected from \cref{eq:betarelaxation}, the second regimes kicks in for $t\gtrsim \tau^*$ when $\beta_t$ relaxes to its steadystate value, $\beta_t \approx |\gamma|^2$. Furthermore, we have $\tau^* = 1/\Delta = \Theta(N)$. The crossover physics is akin to that outlined in \cite{theodoros_criticality} for the setting of continuous sensing at dissipative criticality.

We show a scaling collapse of the radiation QFI by plotting $F(t)/(N^2t^2)$ as a function of rescaled time $\kappa t/N$ in Fig.~\ref{fig:fig2}(b). The perfect collapse confirms that the time crystal shows optimal scaling in both particle number and time $F(t) \sim N^2t^2$ for $t \lesssim \tau^*$, and crosses over to $F(t) \sim N^3t$ for $t \gtrsim \tau^*$ with $\tau^* \propto N$, as also verified by a scaling collapse of $(\beta_t - |\gamma|^2)/N^2 \sim \exp(-t/N)$ (not shown). We comment on the relation of our result to that of Ref.~ \cite{Lesanovsky_BTC}. The rescaling of the dissipation in Ref.~\cite{Lesanovsky_BTC} effectively renders the gap $O(1)$ and the critical point dependent on the system size, thus the QFI crossover is instead from $N^2t^2$ to $N^2 t$ at a $\tau^* = O(1)$.

\textit{Relating to `usual' metrology---}We previously established that a radiation QFI scaling as $\Theta(N^2 T \tau^*)$ is achievable, illustrated with the time crystal example. However, a crucial practical limitation arises: achieving this scaling requires monitoring over timescales $T > \tau^*$, which may not be feasible if the correlation time diverges with system size as $\tau^* \sim N^z$, as is the case in the BTC example with $z = 1$. Moreover, we established a transient quadratic scaling for $\tau \lesssim \tau^*$. This motivates the need to consider protocols that can extract the transient quadratic-in-time scaling of the QFI from the radiation alone, without requiring long-time observations. Drawing inspiration from the proof techniques developed in the theory of noisy quantum metrology~\cite{normbounds1_fujiwara2008fibre,normbounds2_demkowicz2012elusive,normbounds3_demkowicz2014using,normbounds4_demkowicz2017adaptive,normbounds5_zhou2018achieving,normbounds5_sekatski2017quantum,normbounds6_kurdzialek2024universalcombs,normbounds7_kurdzialek2023using}, we derive a sufficient condition under which such scaling can be observed. This approach also provides a bridge between our results and the conventional framework of quantum metrology.

The key algebraic condition for successful noisy metrology can be stated as follows: if the Lindblad span is defined as \(\mathsf{L} = \text{span}\left\{\mathbbm{1}, L_m, L_m^\dagger, L_m^\dagger L_n \right\}_{m,n},\) then high-precision noisy metrology is \textit{not} possible when $H_s \in \mathsf{L}$—a condition referred to as Hamiltonian in Lindblad Span (HLS). As we now demonstrate, this condition can also be leveraged in the context of continuous sensing through monitoring the emitted radiation. In particular, we show that if the HLS condition holds, the corresponding tensor of the output radiation takes the following form [Lemma~\ref{lem:HKSImplies} in \cite{SM}]:

\begin{align} \label{eq:mainVofHKS}
    \quad
    \begin{tikzpicture}[scale=0.5,baseline={([yshift=-0.65ex] current bounding box.center) }]
        \GTensor{0,0}{1.2}{.6}{\small $V$}{7};
  \end{tikzpicture} \quad \equiv 
 \quad
     \begin{tikzpicture}[scale=0.5,baseline={([yshift=-0.65ex] current bounding box.center) }]
         \GTensor{0,0}{1.2}{.6}{\small $W$}{7};
           \RTensor{0,-2}{1.}{.6}{\small $u_\theta$}{1};
   \end{tikzpicture}
       \quad 
       \text{ with }
       \quad 
       \begin{tikzpicture}[scale=0.5,baseline={([yshift=-0.65ex] current bounding box.center) }]
        \GTensor{0,0}{1.2}{.6}{\small $\dot{W}$}{7};
        \GTensor{0,-2}{1.2}{.6}{\small ${\bar{W}}$}{0};
        \draw[very thick,draw=red] (-1.2,-2)--(-1.2,0);
    \end{tikzpicture} \quad 
    = \quad 0 
   \end{align} 
with $u_\theta = e^{\iota\theta h}$ for some Hermitian $h\in \mathbb{C}^{d\times d}$. That is, $\ket{\psi_V} = U_\theta \ket{\psi_W}$ with $U_\theta = \exp{\iota\theta H}$, with a radiation-only single-body Hamiltonian $H = \sum_\tau h_\tau$. The condition on $W$ ensures an exclusively linear QFI, where $\beta^W_\tau = 0~\forall \tau$ and $\gamma^W = 0$, leading to $F[\psi_W] = 4T\cdot\alpha^W$. Then, we have the following result for the QFI of $V$ [Lemma~\ref{lem:HKSQFI} in \cite{SM}]:
\begin{equation}
    F[\psi_V] = 4\var_W[H] + 4T\alpha^W + 8\text{Im}\expect{\psi_W}{H}{\dot{\psi}_W}
\end{equation}
The key insight is that the quadratic contribution to the QFI, $F[\psi_V]$, arises exclusively from the Hamiltonian variance encoded in the radiation. Specifically, the term $\beta^V_\tau$ consists solely of covariance contributions [see Sec.~\ref{subsec:continuousHKS} of SM \cite{SM} for details]. For gapped dynamics, it can be shown that the mixed state of the radiation exhibits the characteristic quadratic scaling in time even in the transient regime $T < \tau^*$, up to sub-linear corrections [see Lemma~\ref{lem:HKSmixedqfibound} in \cite{SM}]. Notably, the HLS condition is naturally satisfied in a broad class of systems relevant for continuous sensing~\cite{Lesanovsky_BTC,theodoros_criticality}. To see this, observe that a generic linear combination of ladder-type jump operators typically generates the collective $\mathfrak{su}(2)$ spin algebra within the span $\mathsf{L}$.

\textit{QFI of long-range TI-MPS.---}
We now turn to channels with multiple steady states. These are nongeneric in finite-size systems, but might be stable in the presence of symmetries. When there are multiple steady states, 
the spectrum of the quantum channel can be divided into a long-lived part that lies on the unit circle [or equivalently, in the Lindbladian case, consists of eigenvalues that have real part $0$], and a dissipative part. If we assume that each long-lived eigenvalue has multiplicity $\mu$, we can express the transfer matrix as
%
%
\begin{equation}
     T_V = \sum_{\Delta,\mu}e^{\iota \Delta}|\Psi_{\Delta\mu} \rrangle \llangle J^{\Delta\mu}| + \tilde{T}_V
\end{equation}
with $|\Lambda(\tilde{T}_V)| < 1$ and a biorthogonal system formed by the eigenvectors, $\llangle J^{\Delta\mu}|\Psi_{\Delta' \nu}\rrangle = \delta^\Delta_{\Delta'} \delta^\mu_\nu$. Here $\Psi_{\Delta\mu}$ forms a basis for the right eigenmatrices and $J^{\Delta\mu}$ are the set of conserved ($\Delta = 0$) or oscillating ($\Delta \neq 0$) quantities of the Liouvillian, as $\text{Tr}(J^{\Delta\mu\dagger}\rho(t)) =  e^{\iota\Delta t}\text{Tr}(J^{\Delta\mu\dagger}\rho_{in})$. In contrast to the injective case, the memory of the initial state is retained in the dynamics through overlaps between the initial state and non-trivial conserved quantities $c_{\Delta\mu} := \llangle J^{\Delta\mu}|\rho_{in}\rrangle$, with the asymptotic state written as \cite{albert2016geometry}
\begin{eqnarray}
    |\rho_{\infty}(t)\rrangle = \sum_{\Delta,\mu} c_{\Delta,\mu} e^{\iota \Delta t}|\Psi_{\Delta\mu}\rrangle
\end{eqnarray}
Our result~\eqref{eq:QFI_TIMPS} extends to this noninjective case, provided the variables $\alpha, \{\beta_\tau\},\gamma$ are redefined as follows, to account for the initial state-dependence:
\begin{align} \label{eq:longrangealpha}
  \alpha \equiv \sum_{\Delta,\mu} c_{\Delta\mu} \quad
\begin{tikzpicture}[scale=0.5,baseline={([yshift=-0.65ex] current bounding box.center) }]
      \GTensor{0,0}{1.2}{.6}{\small $\dot{V}$}{7};
      \GTensor{0,-2}{1.2}{.6}{\small $\dot{\bar{V}}$}{0};
      \draw[very thick,draw=red] (1.2,0)--(2.03,0);
    \DTensor{2,-1}{1.}{.6}{\small $\Psi_{\Delta\mu}$}{1};
          \draw[very thick,draw=red] (1.2,-2)--(2.03,-2);
        \draw[very thick,draw=red] (-2,0)--(-1.2,0);
        \draw[very thick,draw=red] (-2,-2)--(-1.2,-2);
    \UTensor{-2,-1}{1.}{.6}{\small $J^{\Delta\mu}$}{1};
\end{tikzpicture}
    \quad 
\end{align} 

\begin{align} \label{eq:longrangebeta}
  \beta_\tau \equiv \sum_{\Delta,\mu} c_{\Delta\mu} \quad
\begin{tikzpicture}[scale=0.5,baseline={([yshift=-0.65ex] current bounding box.center) }]
      \GTensor{-5,0}{1.2}{.6}{\small $\dot{V}$}{7};
      \GTensor{-5,-2}{1.2}{.6}{\small $\bar{V}$}{0};
          \UTensor{-7,-1}{1.}{.6}{\small $J^{\Delta\mu}$}{1};
    \FUnitary{-3.25,-2.25}{1}{.6}{\small $T_V^\tau$}{5};
    \draw[very thick,draw=red] (-6,-2)--(-7,-2);
        \draw[very thick,draw=red] (-6,0)--(-7,0);
      \GTensor{-1.45,0}{1.2}{.6}{\small ${V}$}{7};
      \GTensor{-1.45,-2}{1.2}{.6}{\small $\dot{\bar{V}}$}{0};
          \DTensor{.5,-1}{1.}{.6}{\small $\Psi_{\Delta\mu}$}{1};
          \draw[very thick,draw=red] (-0.5,-2)--(0.53,-2);
          \draw[very thick,draw=red] (-0.5,0)--(0.53,0);
      \end{tikzpicture}
    \quad
\end{align} 

\begin{align}\label{eq:longrangegamma}
  \gamma \equiv \sum_{\Delta,\mu} c_{\Delta\mu}\quad
\begin{tikzpicture}[scale=0.5,baseline={([yshift=-0.65ex] current bounding box.center) }]
      \GTensor{0,0}{1.2}{.6}{\small $\dot{V}$}{7};
      \GTensor{0,-2}{1.2}{.6}{\small $\bar{V}$}{0};
      \draw[very thick,draw=red] (1.2,0)--(2.03,0);
    \DTensor{2,-1}{1.}{.6}{\small $\Psi_{\Delta\mu}$}{1};
          \draw[very thick,draw=red] (1.2,-2)--(2.03,-2);
        \draw[very thick,draw=red] (-2,0)--(-1.2,0);
        \draw[very thick,draw=red] (-2,-2)--(-1.2,-2);
    \UTensor{-2,-1}{1.}{.6}{\small $J^{\Delta\mu}$}{1};
\end{tikzpicture}
    \quad 
\end{align}

%
\noindent Crucially, $\beta_\infty$ no longer coincides with $|\gamma|^2$: there are two sums over $c_{\Delta \mu}$ in \cref{eq:longrangebeta} and only one in \cref{eq:longrangegamma} This is as one might expect: correlation length is infinite, so in principle $T^2$ scaling of the full QFI is attainable. By the same token, however, the full QFI cannot be separated into a ``boundary'' part that depends on $\mathcal{S}$ and a ``bulk'' part that depends purely on the radiation. 

We now discuss these subtleties in the relatively simple context of an open-system process that generates a GHZ state on the radiation. 
The procedure for sequentially generating (``radiating'') a GHZ state is well known: one initializes the system $\mathcal{S}$ in the $\ket{+}$ state, and at each time step $t$ one applies a CNOT gate from it to the radiation qubit in $\mathcal{H}_t$. After $T$ time steps one has a GHZ state on $T+1$ qubits (one belonging to the system). To convert this protocol into a sensing scheme, one can alternate the CNOT gates with evolution for a time $\delta$ under the sensing Hamiltonian $H_s = \theta Z$, acting on $\mathcal{S}$. For this simple setup one can explicitly write the state $\ket{\Psi; t} = \frac{1}{\sqrt{2}} \left(\ket{0}^{\otimes T + 1} + e^{i T \delta \theta} \ket{1}^{\otimes T + 1}\right)$. Given access to the full state, one can evidently sense $\theta$ to precision $1/T$, so the full QFI scales as $T^2$. It is equally evident that this sensitivity requires access to \emph{every} qubit: the radiation state, achieved by tracing out $\mathcal{S}$, is insensitive to $\theta$. 
As a simple extension of this idea, we can construct a protocol with QFI $\sim N^2 t^2$, by starting with a GHZ state and applying a channel that dephases every site in the computational basis.

Consider a $N-$qubit system driven by $H_s = \omega \sum_i h_i$ with $h_i$ denoting a bounded single-body generator $h$ at qubit $i$, admitting the spectral decomposition $h\ket{\pm E} = \pm E \ket{\pm E}$. We monitor the $\ket{-E}$ population at each site with, $L_i = \sqrt{\kappa}\ket{-E}\bra{-E}$. If we denote the eigenbasis of $H_s\ket{E_i} = E_i\ket{E_i}$ with $-NE \leq E_i \leq NE$, then we see that $\mathcal{L}(\ket{E_i}\bra{E_i}) = 0 ~~\forall i$. The transfer matrix admits the decomposition $T_V = \sum_i |E_i\rrangle \llangle J^i| + \tilde{T}_V$ where $|E_i\rrangle \equiv \ket{E_i}\otimes\ket{E_i}$. Starting in the GHZ state $\ket{\psi} = (\ket{E_M} + \ket{E_m})/\sqrt{2}$ with $\ket{E_{M/m}} = \ket{\pm E}^{\otimes N}$ we see that $\alpha = (E_M^2 + E_m^2)/2$, $\beta_\tau = (E_M^2 + E_m^2)/2$, and $|\gamma|^2 = [(E_M + E_m)/2]^2$. We conclude that $F(t) = (E_M-E_m)^2t^2$. Since $|E_{M/m}| = N|E|$, one has that $F(t) = \Omega(N^2t^2)$.

\emph{Unentangled measurements}.---The QFI quantities we have discussed so far are properties of the emitted \emph{quantum} state. To measure these, in principle, one needs to perform entangling gates between qubits emitted at different times. In practice, the radiation is typically in the form of photons, which must be measured immediately. Therefore, one has access not to the radiation state but to the dephased state $\rho_D$. These instantaneous measurements allow one to reconstruct arbitrary product observables (e.g., Pauli strings), but not entangled observables. It is natural to ask under what conditions an enhanced QFI can be observed through classical temporal correlations in the outcomes of product measurements.  

For the GHZ state described above, this is possible. The optimal observable to detect the phase of $\ket{\Psi,t}$ is $X_s \prod_{t=1}^T X_t$. To reconstruct this observable it suffices to measure each radiated qubit in the $X$ basis as it is emitted. We now observe that this is a consequence of the symmetry structure of the GHZ state, rather than a fine-tuned feature. The transfer matrix for the GHZ state has a ``strong'' conservation law: each Kraus operator commutes with the operator $Z$ on the bond space. Moreover, the radiation from states with distinct eigenvalues of $Z$ records the breaking of the associated $\mathbb{Z}_2$ symmetry: e.g., the $\ket{0}$ state radiates $\ket{0}$'s and vice versa. Together, these properties can be combined into the tensor relations:
\begin{align}\label{eq:Xsymmetry}
  \begin{tikzpicture}[scale=0.5,baseline={([yshift=-0.65ex] current bounding box.center)}]
  \RTensor{0,0}{1.}{.6}{\small $X$}{1};
  \GTensor{0,2}{1.2}{.6}{\small $G$}{7};
  \RTensor{-2,2}{1.}{.6}{\small $X$}{3};
  \RTensor{2,2}{1.}{.6}{\small $X$}{3};
\end{tikzpicture} \quad = \quad 
\begin{tikzpicture}[scale=0.5,baseline={([yshift=-0.65ex] current bounding box.center)}]
  \GTensor{0,2}{1.2}{.6}{\small $G$}{7};
\end{tikzpicture}
\end{align}

\begin{align}\label{eq:phasesymmetry}
  \begin{tikzpicture}[scale=0.5,baseline={([yshift=-0.65ex] current bounding box.center)}]
  \GTensor{0,2}{1.2}{.6}{\small $G$}{7};
  \RTensor{2,2}{1.}{.6}{\small $e^{\iota \theta Z}$}{3};
\end{tikzpicture} \quad = \quad 
\begin{tikzpicture}[scale=0.5,baseline={([yshift=-0.65ex] current bounding box.center)}]
  \GTensor{0,2}{1.2}{.6}{\small $G$}{7};
  \RTensor{-2,2}{1.}{.6}{\small $e^{\iota \theta Z}$}{3};
\end{tikzpicture}
\end{align}
\cref{eq:phasesymmetry} allows one to commute the perturbations $e^{i \theta Z}$ to the initial time, and \cref{eq:Xsymmetry} ``pushes'' a product measurement of $X$ into the virtual space. These relations can be combined to establish the $O(T^2)$ scaling of the full QFI. Importantly, these relations are not fine-tuned properties of the GHZ state; instead, they are generic consequences of the symmetry that gives rise to multiple steady states. In general, the quantity one is sensing might not be the conserved operator; however, provided it has nonzero projection onto that operator, $O(T^2)$ scaling of the full QFI is attainable with product measurements.

\textit{Discussion.---}
In this work we derived compact expressions for the QFI of an open quantum system together with the radiation it emits. We demonstrated that this QFI is enhanced by temporal correlations in the emitted radiation. At least in some specific instances, like symmetry-breaking, these temporal correlations are captured by product operators, and can therefore be reconstructed even if one is constrained to measure each qudit of radiation as it is emitted. Ideally, to guide near-term experiments, one would want an expression for the classical Fisher information of the measurement record from measuring each radiated qudit in an arbitrary basis. Our results upper-bound this quantity, and show that for some states the upper bound is saturated. However, extending the MPS formalism to compute the mixed-state QFI of the dephased state $\rho_D$ (which captures the classical measurement record) requires a further technical advance beyond our results. Whether adaptive measurements, based on classical feedforward schemes~\cite{LOCC_zhou,plenio_decoder,Lesanovsky_BTC}, can always achieve the optimal scaling of QFI is an interesting question for future work. Another interesting question is the robustness of radiation sensing against noise or loss of some fraction of the radiated qubits: naively, any gain from correlations should survive until the loss rate is comparable to the inverse correlation time $1/\tau^*$, but this remains to be established. Finally, metrology might provide a useful framework for treating the stability of pure- and mixed-state phases in a unified way. Our results relate the sensitivity of a steady state on $\mathcal{S}$ in $d$ dimensions, against perturbations of its Lindbladian, to that of a ground state on $\mathcal{S} \otimes \prod_t \mathcal{H}_t$ in $d+1$ dimensions, against perturbations of its parent Hamiltonian---thus potentially allowing one to relate stability results in these two distinct contexts.

\textit{Acknowledgements.---}The authors are grateful to Theodoros Ilias, Benjamin Lev, Hideo Mabuchi, Daniel Malz, and Sai Vinjanampathy for helpful discussions and collaborations on related topics.

\textit{Note added}.---While this work was nearing completion, related papers on tensor-network methods for continuous sensing were posted~\cite{yang2025quantumcramerraoprecisionlimit,khan2025tensornetworkapproachsensing,abbasgholinejad2025theory}. These works use the tensor-network technology to evaluate the dephased QFI numerically.

\bibliographystyle{unsrt}
\bibliography{apssamp}

\onecolumngrid
\renewcommand{\figurename}{Supplementary Figure}
\setcounter{figure}{0}    

\newpage
\section*{Supplementary Material}
\begin{enumerate}[label=\textcolor{red}{\Roman*.}]
    \item \hyperref[sec:TIMPS]{Translationally invariant matrix product states}  
    \item \hyperref[sec:QFI]{Quantum Fisher Information} 
    \begin{enumerate}
      \item[\textcolor{red}{(a)}] \hyperref[subsec:qfidef]{Definitions}
      \item[\textcolor{red}{(b)}] \hyperref[subsec:qfiprop]{Properties}
    \end{enumerate}
    \item \hyperref[sec:hierarcy]{Hierarchy of QFIs in continuous sensing}
    \item \hyperref[sec:qfinormalproof]{QFI of gapped dynamics}
    \begin{enumerate}
      \item[\textcolor{red}{(a)}] \hyperref[subsec:proofmain]{Proof of main result}
      \item[\textcolor{red}{(b)}] \hyperref[subsec:asymprate]{Asymptotic Rate}
      \item[\textcolor{red}{(c)}] \hyperref[subsec:contlimit]{Taking the continuous time limit}
      \item[\textcolor{red}{(d)}] \hyperref[subsec:lowerboundonmixed]{Lower bound on mixed radiation QFI} 
    \end{enumerate}
    \item \hyperref[sec:jointQFI]{Haar averaged QFI}
    \begin{enumerate}
      \item[\textcolor{red}{(a)}] \hyperref[subsec:haarunique]{Unique fixed point}
      \item[\textcolor{red}{(b)}] \hyperref[subsec:haargeneric]{Non-normal MPS}
    \end{enumerate}
    \item \hyperref[sec:statbasisprops]{Properties of the steadystate spectrum}
    \item \hyperref[sec:molmerformula]{Relation to the formula in previous works}
    \item \hyperref[sec:normboundssec]{Noisy Metrology}
    \begin{enumerate}
      \item[\textcolor{red}{(a)}] \hyperref[subsec:continuousHKS]{Continuous sensing under `Hamiltonian in Kraus Span (HKS)'}
    \end{enumerate}
\end{enumerate} 
\section{Translationally-invariant matrix product states}\label{sec:TIMPS}
We have, a translationally-invariant matrix product state (TI-MPS) $\ket{V}$ specified by the tensor $V\in \C^{D\times D\times d}$ depicted as in the main text,
\begin{align}
    \ket{V} =  \dotso \quad
 \begin{tikzpicture}[scale=0.5,baseline={([yshift=-0.65ex] current bounding box.center) }]
         \GTensor{0,0}{1.2}{.6}{\small $V$}{7};
       \GTensor{1*\singledx,0}{1.2}{.6}{\small $V$}{7};
       \GTensor{2*\singledx,0}{1.2}{.6}{\small $V$}{7};
       \GTensor{3*\singledx,0}{1.2}{.6}{\small $V$}{7};
       \GTensor{4*\singledx,0}{1.2}{.6}{\small $V$}{7};
 \end{tikzpicture}
     \quad \dotso
 \end{align}
As in the main text, the isometry is parametrized, $V \equiv V_\theta$, and the parameter dependence is omitted for notational convenience. For an normal TI-MPS, the associated transfer matrix takes the following form
\begin{align}
  T_V := \quad
\begin{tikzpicture}[scale=0.5,baseline={([yshift=-0.65ex] current bounding box.center) }]
      \GTensor{0,0}{1.2}{.6}{\small $V$}{7};
      \GTensor{0,-2}{1.2}{.6}{\small $\bar{V}$}{0};
\end{tikzpicture}
\quad = \quad 
\begin{tikzpicture}[scale=0.5,baseline={([yshift=-0.65ex] current bounding box.center) }]
  \DTensor{2,-1}{1.}{.6}{\small $r$}{1};
\draw[very thick,draw=red] (2,0)--(.5,0);
\draw[very thick,draw=red] (2,-2)--(.5,-2);
\DTensor{4,-1}{1.}{.6}{\small $l$}{1};
\draw[very thick,draw=red] (4,0)--(5.5,0);
\draw[very thick,draw=red] (4,-2)--(5.5,-2);
\end{tikzpicture}  \quad +  \tilde{T}_V
\end{align} 
with $|\Lambda(\tilde{T}_V)| < 1$, and the left and right eigenvectors satisfy the following,
\begin{align}
  \begin{tikzpicture}[scale=0.5,baseline={([yshift=-0.65ex] current bounding box.center) }]
    \DTensor{-2,-1}{1.}{.6}{\small $l$}{1};
        \GTensor{0,0}{1.2}{.6}{\small $V$}{7};
        \GTensor{0,-2}{1.2}{.6}{\small $\bar{V}$}{0};
        \draw[very thick,draw=red] (-2,0)--(-1,0);
        \draw[very thick,draw=red] (-2,-2)--(-1,-2);
  \end{tikzpicture}
      \quad = \quad 
      \begin{tikzpicture}[scale=0.5,baseline={([yshift=-0.65ex] current bounding box.center) }]
        \DTensor{-2,-1}{1.}{.6}{\small $l$}{1};
            \draw[very thick,draw=red] (-2,0)--(-.5,0);
            \draw[very thick,draw=red] (-2,-2)--(-.5,-2);
      \end{tikzpicture}    
      \quad \text{ and } \quad 
      \begin{tikzpicture}[scale=0.5,baseline={([yshift=-0.65ex] current bounding box.center) }]
        \DTensor{2,-1}{1.}{.6}{\small $r$}{1};
            \GTensor{0,0}{1.2}{.6}{\small $V$}{7};
            \GTensor{0,-2}{1.2}{.6}{\small $\bar{V}$}{0};
            \draw[very thick,draw=red] (2,0)--(1,0);
            \draw[very thick,draw=red] (2,-2)--(1,-2);
      \end{tikzpicture}    
      \quad = \quad 
      \begin{tikzpicture}[scale=0.5,baseline={([yshift=-0.65ex] current bounding box.center) }]
        \DTensor{2,-1}{1.}{.6}{\small $r$}{1};
            \draw[very thick,draw=red] (2,0)--(.5,0);
            \draw[very thick,draw=red] (2,-2)--(.5,-2);
      \end{tikzpicture}    
  \end{align} 
and are normalized as follows,
\begin{align}
    \begin{tikzpicture}[scale=0.5,baseline={([yshift=-0.65ex] current bounding box.center) }]
        \DTensor{-2,-1}{1.}{.6}{\small $l$}{1};
        \draw[very thick,draw=red] (-2,0)--(.5,0);
        \draw[very thick,draw=red] (-2,-2)--(.5,-2);
        \DTensor{0.5,-1}{1.}{.6}{\small $r$}{1};
      \end{tikzpicture}    \quad = \quad 1.
\end{align}

For the case of $V = \sum_m K_m \otimes \ket{m}$ (with $K_m \equiv K_m^\theta$ encoding the parameter) the corresponding channel is $\mathcal{E}(\cdot) := \sum_m K_m(\cdot) K_m^\dagger$. If the channel admits a unique fixed point $\rho_{ss}$, one has that $|r\rrangle = |\rho_{ss}\rrangle$ and $\llangle l| = \llangle \mathbbm{1}|$. The linear maps for performing left and right gauge transformations, denoted $L$ and $R$ respectively, can be obtained from the left and right eigenvectors of the transfer matrix as $l = L^\dagger L$ and $r= R^\dagger R$. Specifically, one defines the left- and right-canonical tensors for $V^m \equiv \bra{m}V$ as \cite{Vanderstraeten_2019_tangentspace},
\begin{eqnarray}
    V^m_L := LV^mL^{-1},~V_R^m:= R^{-1}V^mR ~~\forall m\in[d].
\end{eqnarray}
Furthermore, the tensor at the center of orthogonality for the mixed gauge is given as $V^m_C = LV^mR$.

\section{Quantum Fisher Information} \label{sec:QFI}
\subsection{Definitions}\label{subsec:qfidef}
All Fisher information quantities can be described quite generally in the following setting. Consider a smooth encoding of a $k-$dimensional parameter $\theta \in \Theta \subseteq \mathbb{R}^k$ in a Hilbert space $\mathcal{H}$ with a smooth function $\rho_\theta:\Theta \to D(\mathcal{H})$. The task of quantum sensing requires a sharp curvature of the \textit{pullback} distance \cite{Fisher_NISQ_Meyer} $d(\theta,\theta') := d_f(\rho_\theta,\rho_{\theta'})$. The central idea is that one only has access to objects in the state space $D(\mathcal{H})$ encoding the parameter $\theta$, a sharp curvature in this space upon a perturbation $\theta \to \theta + \delta$ implies the possibility of high-precision sensing. We now formalize this intuition. Consider the Taylor expansion of the induced distance $d(\theta,\theta+\delta)$ for $||\delta||_2 \ll 1$
\begin{equation}
    d(\theta,\theta+\delta) = \frac{1}{2}\delta^\top M_\theta \delta + O(|\delta|^3).
\end{equation}
Terms up to the first order vanish because $d(\cdot,\cdot) \geq 0$ and $d(\rho,\rho) = 0$. The curvature is thus encoded in the real and symmetric \textit{information matrix} $M_\theta \in \mathbb{R}^{k \times k}$ defined as 
\begin{equation} \label{eq:infomatrix}
    M_\theta ^{ij}  \equiv \frac{\partial}{\partial\delta_i}\frac{\partial}{\partial\delta_j} d(\theta,\theta+\delta)~\mid _{\delta = 0}
\end{equation}
For instance, if the encoding is classical $\rho_\theta \equiv \sum_m p_\theta(m)\ket{m}\bra{m}$, and one employs the Kullback-Liebler divergence $d_{KL}(p||q) = \sum_m p_m \log{p_m/q_m}$ as the distance function, one obtains the classical Fisher information matrix 
\begin{equation} \label{eq:cfimatrix}
    I_\theta ^{ij} = \sum_m \frac{1}{p_\theta(m)} \frac{\partial p_\theta(m)}{\partial \theta_i}\frac{\partial p_\theta(m)}{\partial \theta_j} = \mathbb{E}_m\left[\partial_i \log{p_\theta(m)}\partial_j \log{p_\theta(m)}\right]
\end{equation}
The classical Fisher information is unique up to constant factors as long as the distance function chosen is monotonic, i.e., it satisfies $d(T[p]||T[q]) \leq d(p||q)$ under a stochastic matrix $T$. \cite{morozova1991markov}

The uniqueness is no longer true for the general quantum case, and we stick to a particular definition which agrees with standard distance notions in two cases. 

\begin{defn}[QFI]
    For a smooth encoding $\rho_\theta: \Theta \subseteq \mathbb{R}^k \to D(\mathcal{H})$, define the symmetric logarithmic derivative (SLD) operators for each $i \in [k]$ implicitly as follows, 
    \begin{equation}
        \frac{\partial}{\partial\theta_i}\rho_\theta =: \frac{1}{2}\{\rho_\theta, L_i\}
    \end{equation}
    Then, the QFI matrix is defined as, 
    \begin{equation}
        F_{\theta}^{ij}=\tr(\rho_\theta L_i L_j)  ~~ \forall (i,j) \in [k]^2
    \end{equation}
\end{defn}
For an arbitrary parametrized mixed state $\rho_\theta$, the quantum Fisher information matrix with respect to the parameter $\theta$ is defined as 

where $L_i$ are the set of symmetric logarithmic derivative (SLD) operator, defined implicitly as 

This definition coincides with the pullback metric induced by the Bures distance $d(\rho,\sigma) = 2(1-\tr[\sqrt{\sqrt{\rho}\sigma\sqrt{\rho}}])$ when the perturbation remains full rank, as well as under the  distance $d(\psi,\phi) = 2(1-|\inner{\psi}{\phi}|^2)$ for pure states. The SLD operators can be thought of as non-commutative generalizations of the derivatives in \cref{eq:cfimatrix}.

The reason for why the curvature \cref{eq:infomatrix} is an information-theoretic quantity is well encapsulated by the Cram\'{e}r-Rao theorem, establishing the clear operational relevance.

\begin{theorem}[\cite{Caves_QFI}]
    Consider the task of estimating $\theta\in\Theta\subseteq \mathbb{R}^k$ from the measurement channel $\Pi = \{\Pi_m\}_{m=1}^{p}$ employed on the smoothly parametrized state $\rho_\theta$. Denote $p_\theta(m) := \tr[\Pi_m \rho_\theta]$. Then, for any locally unbiased estimator operating on $s$ measurement shots $\hat{\theta}:[p]^s \mapsto \Theta$ with $\mathbb{E}[\hat{\theta}] = \theta$ one has that 
    \begin{eqnarray}
        (\Delta \hat{\theta})^2 \geq \frac{1}{s} I_\theta^{-1} \geq \frac{1}{s} F_\theta^{-1}
    \end{eqnarray}
    where $I_\theta$ and $F_\theta$ are the classical and quantum Fisher information matrices, $(\Delta\hat{\theta})^2$ is the covariance matrix with $[(\Delta \hat{\theta})^2]_{ij} := \mathbb{E}[\Delta\hat{\theta}_i\Delta\hat{\theta}_j]$ and $A \geq B$ iff $A-B$ is positive semidefinite.
\end{theorem}

\noindent 
In the single parameter case, there always exists a measurement such that the CFI of the measurement outcomes $I_\theta$ is equal to the QFI $F_\theta$. The optimal measurement can be carried out by measuring in the eigenbasis of $L_\theta$. This is no longer true for the multiparameter case, as the different SLD operators need not commute. From here onwards, we stick to the problem of a single real-valued encoded parameter $\theta \in \mathbb{R}$.

\noindent 
Moreover, for a pure parametrized state encoding a single parameter $\theta \in \mathbb{R}$, $\rho_\theta \equiv \ket{\psi_\theta}\bra{\psi_\theta}$, we have the following form of the QFI \cite{QFI_Multiparameter_Review,metrology_QIS_},
\begin{equation}
    F[\ket{\psi_\theta}] = 4\left(\inner{\partial_\theta\psi_\theta}{\partial_\theta\psi_\theta} - |\inner{\psi_\theta}{\partial_\theta\psi_\theta}|^2\right)
\end{equation}
which is the central formula for the MPS calculations in Sec.~\ref{sec:qfinormalproof}. Upon considering the state $\ket{\psi_\theta(t)} = e^{-\iota H \theta t}\ket{\psi}$, one has that $F(t) = 4 t^2 \var_{\ket{\psi}(t)}[H] \leq t^2 (E_M - E_m)^2$, where $E_{M(m)}$ are the maximum(minimum) eigenvalues of $H$. For the more general case of a sensing and control Hamiltonian, one has $H_\theta = \theta H_s + H_c$. Here, $H_s$ represents the sensing Hamiltonian, and $H_c$ accounts for any additional controlled or unwanted interactions. We define a effective Hamiltonian as follows \cite{dynamical_puig2024dynamicalsteadystatemanybodymetrology},
\begin{equation}
    H_{\text{eff}} = \int_0^1 e^{-\iota ts H_\theta}H_s e^{\iota ts H_\theta}ds
\end{equation}
Then, a similar result holds and we have 
\begin{equation}
    F = 4t^2\var_{\ket{\psi(t)}}[H_{\text{eff}}]
\end{equation}
It again follows that $F \leq t^2 (E_M - E_m)^2$, with the bound depending on the spectrum of $H_s$ as before. 

\subsection{Properties} \label{subsec:qfiprop}
We state some known properties of the QFI \cite{QFI_Multiparameter_Review,Fisher_NISQ_Meyer,metrology_QIS_}
\begin{enumerate}
    \item Positivity. $F[\rho_\theta] \geq 0$, with equality iff $\partial_\theta\rho_\theta = 0$.
    \item The QFI is convex.
    \begin{equation}
        F[p\rho_1 + (1-p)\rho_2] \leq pF[\rho_1]+ (1-p)F[\rho_2]
    \end{equation}
    Furthermore, if the distribution $\{p_x\}$ also depends on $\theta$, then an extended version of this inequality is \cite{extendedconvexity}
    \begin{equation}
        F[\sum_x p_x\rho_x] \leq I[p_x] + \sum_x p_x F[\rho_x]
    \end{equation}
    \item The QFI is additive on product states, 
    \begin{equation}
        F[\otimes_{i}\rho_i] = \sum_i F[\rho_i]
    \end{equation}
    \item The QFI is monotonically non-increasing under (parameter-independent) CPTP maps. 
    \begin{equation}
        F[\mathcal{N}(\rho)] \leq F[\rho]
    \end{equation}
\end{enumerate}

\noindent 
We note that the convexity and product additivity of the QFI implies that any separable state on $N$ particles has QFI $O(N)$. To see this, let $\rho = \sum_i p_i \rho_i^1 \otimes \dots \otimes \rho_i^N$ be the state, and note that $F(\rho) \leq \sum_i p_i F_i$ where $F_i := F(\rho_i^1 \otimes \dots \otimes \rho_i^N)$. Now, note that $F_i = O(N)$ by additivity, and let $f := \max_i F_i/N = O(1)$. Then, $F(\rho) \leq f N = O(N)$. Thus, a separable state cannot show quadratic QFI. 
\section{Hierarchy of QFIs in continuous sensing} \label{sec:hierarcy}
In the most general sequential generation problem, we define isometries $V_{[i]}: \mathcal{S} \to \mathcal{S} \otimes \mathcal{H}_d$ for time-bin $t_i$ with $\mathcal{H}_D \cong \mathbb{C}^D, \mathcal{H}_d \cong \mathbb{C}^d$ to formally describe the emission of a time-binned photon in $t\in [t_i,t_{i+1}]$ with $t_{i+1} - t_i =: dt$. We also denote $T := t/dt$. Each isometry is parametrized, $V_{[i]} \equiv V^{\theta}_{[i]}$.     This leads to the joint system-environment state,

\begin{equation}
    \ket{\Psi;T} := V_{[T]}\dots V_{[2]}V_{[1]} \ket{\psi_0}
\end{equation}
with $\ket{\Psi;T} \in \mathcal{S} \otimes \mathcal{H}_d^{\otimes T}$ describes the joint state of the atom-cavity-photons until time $T$. Given the Kraus operators $\{K_m\}_{0 \leq m\leq d-1}$, we have
\begin{equation}
    \ket{\Psi;T} = \mathlarger{\sum}\limits_{m_i \in [d],}K_{m_T}\dots K_{m_1}\ket{\psi_0} \otimes \ket{m_1\dots m_T}
\end{equation}
Rewriting the previous expression compactly,
\begin{equation}
    \ket{\Psi;T} = \sum\left(K_{\mathbf{m}}\ket{\psi_0}\ket{\mathbf{m}}\right)
\end{equation}
with $K_{\mathbf{m}} := K_{m_T}K_{m_{T-1}}\dots K_{m_1}$.  We now denote $\ket{\psi_\mathbf{m}}:=K_{\mathbf{m}}\ket{\psi_{ac}(0)}/\sqrt{p_\mathbf{m}}$ with $p_\mathbf{m} := \expect{\psi_0}{K_{\mathbf{m}}^\dagger K_{\mathbf{m}}}{\psi_0}$ and write the state as
\begin{eqnarray}
    \ket{\Psi;T} = \sum_{\mathbf{m}}\sqrt{p_\mathbf{m}}\ket{\psi_\mathbf{m}} \otimes \ket{\mathbf{m}}
\end{eqnarray}
This is the `parent' state that we work with, and now outline a hierarchy of Fisher informations arising from this.
\begin{itemize}
    \item QFI of the entire (system and radiation) state $F_{SR} \equiv F(\ket{\psi;T};\theta)$. This is also referred to as the `joint QFI.'
    \item QFI of the system
    \begin{eqnarray}
        \rho_S := \text{Tr}_R(\ket{\psi;T}\bra{\psi;T}) = \sum_\mathbf{m}p_\mathbf{m} \ket{\psi_\mathbf{m}}\bra{\psi_\mathbf{m}}
    \end{eqnarray}
    This is pertinent to the study of noisy metrology $F_S \equiv F(\rho_S;\theta)$.
    \item The QFI of the pure radiation state,
    \begin{eqnarray}
        \ket{R} = \sum_{\mathbf{m}} \text{Tr}[K_\mathbf{m} B]  \ket{\mathbf{m}} 
    \end{eqnarray}
    where $B$ is a rank-one boundary operator. This is the TI-MPS state considered in the main text, denoted $F_R$.
    \item QFI of a post-selected system state, $F_\mathbf{m} = F(\ket{\psi_\mathbf{m}})$
    \item The classical Fisher information upon fixing a basis, e.g. dephasing, $\rho_D =  \sum_\mathbf{m}p_\mathbf{m}\ket{\mathbf{m}}\bra{\mathbf{m}}$, with the diagonal `classical' distribution, resulting in $F(\rho_D) \equiv I(p_\mathbf{m})$.
    \item Dephased QFI, where we dephase the photon legs in some basis, e.g., by applying a measurement channel $\rho \mapsto \sum_x \tr(\Pi_x\rho) \ket{x}\bra{x}$. For photocounting, we have $\Pi_\mathbf{m} = \ket{\mathbf{m}}\bra{\mathbf{m}}$, we thus have 
    \begin{equation}
       \varrho_D =  \mathcal{D}(\ket{\Psi}\bra{\Psi}) = \sum_\mathbf{m} p_\mathbf{m} \ket{\psi_\mathbf{m}}\bra{\psi_\mathbf{m}} \otimes \ket{\mathbf{m}}\bra{\mathbf{m}}
    \end{equation}
    Note that the $\psi_\mathbf{m}$ encodes $\theta$ still, but $\mathbf{m}$ does not except for through the probabilities $p_\mathbf{m}$. It is easy to check that the QFI of this object is, 
    \begin{equation}
       F(\varrho_D) = F(\rho_D) + \sum_\mathbf{m} p_\mathbf{m}F_\mathbf{m} =  I(p_\mathbf{m}) + \left(\sum_\mathbf{m} p_\mathbf{m} F(\ket{\psi_\mathbf{m}})\right)
    \end{equation}
    This can be achieved by photocounting, and then performing measurement conditioned on the photon data on the system.
\end{itemize}
Now, we note the following consequences 
\begin{enumerate}
    \item Owing to convexity $F(\rho_S) \leq \sum_\mathbf{m} p_\mathbf{m}F(\ket{\psi_\mathbf{m}})$. This bound is generically loose, as the RHS can be $\Omega(t^2)$, but the LHS is generically $O(t)$. 
    \item Owing to monotonicity, $F_S \leq F_{SR}$ and $F_R \leq F_{SR}$. This implies that if $F_{SR} = O(t)$, then $F_S = O(t)$.
    \item Both of the points above lead to $F(\rho_S) \leq F(\varrho_D)$ and $F(\rho_D) \leq F(\varrho_D)$
    \item By monotonicity again $F(\varrho_D) \leq F_{SR}$, and the cases where this inequality is (parametrically) saturated is of great interest. For instance, that is the case for the GHZ emitter upon dephasing the photons in the $X-$basis.
\end{enumerate}
The results of Sec.~\ref{sec:qfinormalproof} and Sec.~\ref{sec:jointQFI} show that if the underlying quantum dynamics is mixing with spectral gap $\Delta$, then $F_{SR}/t = F_R/t + O(e^{-\Delta t})$. This, along with monotonicity of the QFI implies that the QFI rate of noisy metrology is upper bounded by $f_0 + f_c$ as well, with $f_0,f_c$ as defined in the main text. This fact may be complementary to the bounds on noisy metrology, without computing operator norms \cite{normbounds1_fujiwara2008fibre,normbounds3_demkowicz2014using,normbounds5_zhou2018achieving}.
\section{QFI of gapped dynamics}
\label{sec:qfinormalproof}
\subsection{Proof of main result} \label{subsec:proofmain}
We state the result of the emission QFI generated by dynamics with a unique fixed point. One iteration of the sequential generation step is described by,
\begin{equation}
    V\ket{\psi}_D := \sum_{m=0}^{d-1} {K_m\ket{\psi}_D \otimes \ket{m}_d}
\end{equation}
where $V^{\dagger}V = 1$ is an isometry given as,
\begin{eqnarray}
    V := \sum_{m=0}^{d-1} K_m \otimes \ket{m} \in \C^{D\times D\times d}
\end{eqnarray}

Now, we define the relevant variables in terms of the isometry $V$, the derivative $\partial_\theta V$, and the channel fixed point $T_V|\rho_{ss}\rrangle = |\rho_{ss}\rrangle$ as,

\begin{align}\label{eq:alphanormalappdx}
    \alpha := \quad
  \begin{tikzpicture}[scale=0.5,baseline={([yshift=-0.65ex] current bounding box.center) }]
        \GTensor{0,0}{1.2}{.6}{\small $\dot{V}$}{7};
        \GTensor{0,-2}{1.2}{.6}{\small $\dot{\bar{V}}$}{0};
        \draw[very thick,draw=red] (1.2,0)--(2.03,0);
      \DTensor{2,-1}{1.}{.6}{\small $\rho_{ss}$}{1};
        \draw[very thick,draw=red] (1.2,-2)--(2.03,-2);
        \draw[very thick,draw=red] (-2,0)--(-1.2,0);
        \draw[very thick,draw=red] (-2,-2)--(-1.2,-2);
        \draw[very thick,draw=red] (-2,-2)--(-2,0);
  \end{tikzpicture}
      \quad 
  \end{align} 
  
  \begin{align}\label{eq:betasnormalappdx}
    \beta_\tau := \quad
  \begin{tikzpicture}[scale=0.5,baseline={([yshift=-0.65ex] current bounding box.center) }]
        \GTensor{-5,0}{1.2}{.6}{\small $\dot{V}$}{7};
        \GTensor{-5,-2}{1.2}{.6}{\small $\bar{V}$}{0};
          \draw[very thick,draw=red] (-5-1.15,-2)--(-5-1.15,0);
          \FUnitary{-3.25,-2.25}{1}{.6}{\small $T_V^\tau$}{5};
        \GTensor{-1.45,0}{1.2}{.6}{\small ${V}$}{7};
        \GTensor{-1.45,-2}{1.2}{.6}{\small $\dot{\bar{V}}$}{0};
            \DTensor{.5,-1}{1.}{.6}{\small $\rho_{ss}$}{1};
            \draw[very thick,draw=red] (-0.5,-2)--(0.53,-2);
            \draw[very thick,draw=red] (-0.5,0)--(0.53,0);
        \end{tikzpicture}
      \quad \forall \tau \geq 0
  \end{align} 
  
  \begin{align}\label{eq:gammanormalappdx}
    \gamma := \quad
  \begin{tikzpicture}[scale=0.5,baseline={([yshift=-0.65ex] current bounding box.center) }]
        \GTensor{0,0}{1.2}{.6}{\small $\dot{V}$}{7};
        \GTensor{0,-2}{1.2}{.6}{\small $\bar{V}$}{0};
        \draw[very thick,draw=red] (1.2,0)--(2.03,0);
      \DTensor{2,-1}{1.}{.6}{\small $\rho_{ss}$}{1};
            \draw[very thick,draw=red] (1.2,-2)--(2.03,-2);
          \draw[very thick,draw=red] (-2,0)--(-1.2,0);
          \draw[very thick,draw=red] (-2,-2)--(-1.2,-2);
          \draw[very thick,draw=red] (-2,-2)--(-2,0);
  \end{tikzpicture}
      \quad .
  \end{align} 

Then, we have the following result.
\begin{result}[QFI of emitted multi-photon state] \label{re:qfinormalmps}
    For any CPTP $\mathcal{E}$ map on $\mathcal{H}_D \simeq \C^D$ with the Kraus operators $\{K_m\}_{m=0}^{d-1}$ satisfying $\sum_{m=0}^{d-1}K_m^\dagger K_m = \mathbbm{1}_D$, we define the corresponding isometry, 
    \begin{equation}
        V\ket{\psi}_D := \sum_{m=0}^{d-1} {K_m\ket{\psi}_D \otimes \ket{m}_d} ~.
    \end{equation}
    If $\mathcal{E}$ has a fixed point, viz. $\mathcal{E}(\rho_{ss}) = \rho_{ss}$, then the quantum Fisher information of the corresponding TI-MPS is given as
    \begin{equation}
        F(t;V,\theta) = 4\left(T\alpha + 2\sum_{\tau =0}^{T-2}(T-\tau-1)\beta_\tau - T^2|\gamma|^2\right), 
    \end{equation}
    where $\alpha,\{\beta_\tau\}_\tau,\gamma$ are defined as in \cref{eq:alphanormalappdx,eq:gammanormalappdx,eq:betasnormalappdx} and $T = t/dt$ is an integer.
\end{result}
    
Hereon in this proof, we assume a unit of time $dt = 1$ for notational convenience. The number of discrete time-steps is $t/dt$, hence upon substituting $t\mapsto t/dt$ in the RHS of \eqref{eq:QFI_TIMPS} and noting that $\alpha,\beta,|\gamma|^2$ are $O(dt^2)$, one gets $F(t) = \sum dt(\dots) \to \int_{0}^t(\cdots)$.

The central idea is to exploit the degree of freedom in moving the orthogonality center of the TI-MPS in the mixed gauge, and use translation invariance.

\begin{proof}
    The TI-MPS state is given as 
    \begin{align} \label{eq:stateVappdx}
        \ket{V} =  \dotso \quad
     \begin{tikzpicture}[scale=0.5,baseline={([yshift=-0.65ex] current bounding box.center) }]
             \GTensor{0,0}{1.2}{.6}{\small $V$}{7};
           \GTensor{1*\singledx,0}{1.2}{.6}{\small $V$}{7};
           \GTensor{2*\singledx,0}{1.2}{.6}{\small $V$}{7};
           \GTensor{3*\singledx,0}{1.2}{.6}{\small $V$}{7};
           \GTensor{4*\singledx,0}{1.2}{.6}{\small $V$}{7};
     \end{tikzpicture}
         \quad \dotso
     \end{align}
     Thus, the derivative $\ket{\dot{V}} \equiv \partial\ket{V}/\partial\theta$ is,
\begin{align}\label{eq:delVappdx}
\ket{\partial_\theta V} = \mathlarger{\sum}_t    \quad\dots
\begin{tikzpicture}[scale=0.5,baseline={([yshift=-0.65ex] current bounding box.center) }]
        \GTensor{0,0}{1.2}{.6}{\small $V_L$}{7};
      \GTensor{1*\singledx,0}{1.2}{.6}{\small $V_L$}{7};
      \draw (1*\singledx,-1.8) node {\small $m_{t-1}$};
      \GTensor{2*\singledx,0}{1.2}{.6}{\small $\dot{V}_{C}$}{7};
      \draw (2*\singledx,-1.8) node {\small $m_{t}$};
      \GTensor{3*\singledx,0}{1.2}{.6}{\small $V_R$}{7};
      \draw (3*\singledx,-1.8) node {\small $m_{t+1}$};
      \GTensor{4*\singledx,0}{1.2}{.6}{\small $V_R$}{7};
\end{tikzpicture}
    \quad \dotso
\end{align} 
Now, the expression for the QFI of a pure state $\ket{V}$ is (see Sec.~ for details)
\begin{eqnarray}
    F = 4\left(\underbrace{\langle\partial_\theta V|\partial_\theta V\rangle}_{a~=~d + n.d.} - |\underbrace{\langle\partial_\theta V| V\rangle}_{s}|^2\right)
\end{eqnarray}
We divide the entire expression into two parts, the additive $a$ and the subtracted $|s|^2$ term, with $a = \langle\partial_\theta V|\partial_\theta V\rangle$ and $s = \langle\partial_\theta V| V\rangle$. We further divide first part into diagonal $d$ and non-diagonal $n.d.$ terms, viz., $a = d ~+~ n.d.$, which we now clarify. The expression for $\ket{\partial_\theta V} = \sum_{\tau}\ket{\partial_\theta^\tau V}$ consists of $t$ terms, where we denote the derivative at site $\tau$ by $\partial_\theta^\tau$. Thus, $a = \langle\partial_\theta V|\partial_\theta V\rangle$ consists of $t^2$ terms. We decompose this into $t$ terms corresponding to diagonal elements, where the derivative is at the same $\tau_1 = \tau_2$, and $t^2-t$ non-diagonal elements, corresponding to $\tau_1 \neq \tau_2$. Owing to the same orthogonality center of the $\tau_1=\tau_2$ terms \eqref{eq:delVappdx}, all the diagonal elements are equal $\inner{\partial_\theta^{\tau_1} V}{\partial_\theta^{\tau_1} V} = \alpha$ and contribute $t \times \alpha$. Now, consider the case of $\tau_1\neq \tau_2$. We first note that the sum is reflection-symmetric, that is, $\inner{\partial_\theta^{\tau_1} V}{\partial_\theta^{\tau_2} V} = \inner{\partial_\theta^{\tau_2} V}{\partial_\theta^{\tau_1} V}$. Further, all the non-diagonal terms with the same value of $\tau \equiv |\tau_1 - \tau_2|$ are equal, as the boundary tensors contract to the identity owing to the choice of gauge. Thus, we denote $\beta_\tau := \inner{\partial_\theta^{\tau_1} V}{\partial_\theta^{\tau_1+\tau} V}$. Now, for each $0 \leq \tau \leq t-2$, the degeneracy factor is $(t-1-\tau)$. Thus, the non-diagonal terms contribute a net total of $2\sum_{\tau =0}^{t-2}(t-\tau-1)\beta_\tau$ where the factor of $2$ comes from the reflection symmetry of the sum. Finally, we evaluate the subtracted term $|s|^2$. Note that $s$ consists of $t-$terms, one for each in \eqref{eq:delVappdx}, contracted with the MPS $\ket{V}$, that is, $s = \sum_\tau \inner{\partial_\theta^{\tau} V}{V}$. In the $\tau$-term of this sum, we orthogonalize the MPS $\ket{V}$ at site $\tau$, resulting in the contribution $\gamma$. Owing to translational invariance, we thus have $|s| = |t\gamma|$. Combining all these contributions, we reach the result stated
\begin{equation}
    F(t;V,\theta) = 4\left(t\alpha + 2\sum_{\tau =0}^{t-2}(t-\tau-1)\beta_\tau - t^2|\gamma|^2\right)
\end{equation}

\end{proof}

\subsection{Asymptotic rate}\label{subsec:asymprate}
We have the QFI, 
\begin{equation}
        F(t;V,\theta) = 4\left(T\alpha + 2\sum_{\tau =0}^{T-2}(T-\tau-1)\beta_\tau - T^2|\gamma|^2\right), 
\end{equation}
Now, we have that $\beta_\tau$ relaxes to $|\gamma|^2$ in the following fashion, $\beta_\tau - |\gamma|^2 = \sum_\mu K_\mu e^{-\tau/\tau_\mu}$, where $K_\mu$ are time-independent constants, potentially depending on the bond-space particle number. Specifically, if we define the following two vectors in the doubled space, 
\begin{align}\label{eq:ABmatrix}
    \llangle A|\equiv \quad
   \begin{tikzpicture}[scale=0.5,baseline={([yshift=-0.65ex] current bounding box.center) }]
         \GTensor{0,0}{1.2}{.6}{\scriptsize $\dot{V}$}{7};
         \GTensor{0,-2}{1.2}{.6}{\scriptsize $\dot{\bar{V}}$}{0};
         \draw[very thick,draw=red] (1.2,0)--(2.03,0);
             \draw[very thick,draw=red] (1.2,-2)--(2.03,-2);
           \draw[very thick,draw=red] (-2,0)--(-1.2,0);
           \draw[very thick,draw=red] (-2,-2)--(-1.2,-2);
           \draw[very thick,draw=red] (-2,-2)--(-2,0);
   \end{tikzpicture}
       \quad ~ |B\rrangle \equiv \quad
       \begin{tikzpicture}[scale=0.5,baseline={([yshift=-0.65ex] current bounding box.center) }]
             \GTensor{0,0}{1.2}{.6}{\small $V$}{7};
             \GTensor{0,-2}{1.2}{.6}{\scriptsize $\dot{\bar{V}}$}{0};
             \draw[very thick,draw=red] (1.2,0)--(2.03,0);
                 \draw[very thick,draw=red] (1.2,-2)--(2.03,-2);
               \draw[very thick,draw=red] (-2,0)--(-1.2,0);
               \draw[very thick,draw=red] (-2,-2)--(-1.2,-2);
             \DTensor{2,-1}{1.}{.6}{\small $\rho_{ss}$}{1};
       \end{tikzpicture}
   \end{align} 

then we have $K_\mu = \llangle A|\Psi_\mu\rrangle \llangle J^\mu | B\rrangle$, where $\Psi_\mu$ and $J^\mu$ are the right and left eigenvectors of the Liouvillian respectively. The correlation length is then $\tau^* \equiv \max_\mu \tau_\mu$. Then, we have the following calculation accounting for the $K_\mu$ terms,

\begin{align}
    \left(\alpha T + 2\sum_{\tau=0}^{T-2}(T-1-\tau)(|\gamma|^2 + \sum_\mu K_\mu e^{-\tau/\tau_\mu}) - T^2|\gamma|^2\right)  \\ 
   = \left(\alpha T +  (T^2 - T)|\gamma|^2 + 2\sum_\mu K_\mu\sum_{\tau=0}^{T-2}(T-1-\tau)e^{-\tau/\tau_\mu} - T^2|\gamma|^2\right)  \\ 
    = (\alpha - |\gamma|^2)T + 2\sum_\mu K_\mu\sum_{\tau=0}^{T-2}(T-1-\tau)e^{-\tau/\tau_\mu} \\ 
   = (\alpha - |\gamma|^2)T + 2\sum_\mu K_\mu \sum_{\tau=1}^{T-1}\tau e^{-(T-1-\tau)/\tau_\mu} \\ 
    = (\alpha - |\gamma|^2)T + 2\sum_\mu K_\mu e^{-(T-1)/\tau_\mu}\sum_{\tau=1}^{T-1}\tau e^{\tau/\tau_\mu} \\ 
    =  (\alpha - |\gamma|^2)T + 2\sum_\mu K_\mu \frac{e^{1/\tau_\mu}}{e^{1/\tau_\mu}-1}\left((T-1) + \left[\frac{e^{(1-T)/\tau_\mu} - 1}{e^{1/\tau_\mu} - 1}\right]\right)  
\end{align}
Thus we have upon ignoring constant terms and approximating $1/(1-e^{-1/x})\approx x$ for large $x$ (with $\text{Re}(x) > 0$), 
\begin{equation} 
    F(t;V,\theta) \simeq 4(\alpha-|\gamma|^2) T  + 4\left(2\sum_\mu K_\mu \tau_\mu\right)  T \label{eq:asymptoticrate1}
\end{equation}
Let $K^* = \max_\mu K_\mu$. In general, the $K_\mu \tau_\mu $ term can be complex, but the conjugate pairs also appear in the sum, so one can treat the real part of the Liouvillian eigenvalues as the decay rates $1/\tau_\mu$. Now, if the bond-space encoding Hamiltonian were consisting of one-body terms, then we have $K^* = \Theta(N^2)$. Moreover, if $\tau^*$ dominates the other timescales (and is the dominating correlation length), then we have 
\begin{equation}
    F(t;V,\theta) \simeq \underbrace{4(\alpha-|\gamma|^2)}_{f_0} T  + \underbrace{8K^*\tau^*}_{f_c}  T
\end{equation}

\subsection{Taking the continuous time limit}\label{subsec:contlimit}

The QFI of Result~\ref{re:qfinormalmps} is generally applied to any setup of continuous sensing with a  time-translational invariant dynamics, with appropriate modifications for degeneracy as in the main text. Now, we specialize to the case where the underlying dynamics in a single time step is generated from a Markovian evolution described by 
\begin{eqnarray}
    \dot{\rho} = -\iota [H,\rho] + \sum_{n\geq 1}\left(L_n\rho L_n^{\dagger} - \frac{1}{2}\{L_n^\dagger L_n,\rho\}\right) \equiv \mathcal{L}\rho, \label{eq:LMEappdx}
\end{eqnarray}
where $\mathcal{L}$ is the Liouvillian superoperator. With discretization $dt$, the Kraus operators for the channel $\mathcal{E} \equiv \exp{\mathcal{L}dt}$ are, 
\begin{align} \label{eq:krausopappd1}
    K_0 &= I - \iota Hdt - \frac{1}{2}\sum_n L_n^\dagger L_n dt\\ \label{eq:krausopappd2}
    K_m &= \sqrt{dt}L_m~ \forall m \ge 1  
\end{align}
which ensures that $\sum_{m \geq 0} K_m\rho K_m^{\dagger} = \exp{\mathcal{L}dt}\rho + O(dt^2)$. The continuous time limit is achieved by taking $dt \to 0$.

\textbf{Hamiltonian Sensing.} If the parameter $\theta$ is encoded solely in the Hamiltonian then we note that $\alpha$ and $|\gamma|^2$ are $O(dt^2)$. The $dt$ dependence of $\beta_\tau$ is as follows: both the derivative terms contribute a $dt$ factor, leading to a $dt^2$. However, the transfer matrix $\tilde{T}_V$ has eigenvalues $-dt \lambda_\mu$, where $-\lambda_\mu$ are the eigenvalues of the Liouvillian $\mathcal{L}$. To clarify this $dt$ dependence, we can consider the asymptotic QFI expression \cref{eq:asymptoticrate1}. We see that $\alpha, |\gamma|^2, K_\mu$ are all $O(dt^2)$. Further, $-dt \lambda_\mu \equiv 1/\tau_\mu \implies \tau_\mu \sim 1/dt$. Finally, recall that $T = t/dt$. Hence, the $f_0$ term is finally $\sim dt$, and the $f_c$ term is independent of $dt$. Taking the continuous time limit $dt \to 0$, only the $f_c$ term survives. This is consistent with the perturbation expansion of Ref.~\cite{molmer_14}, as we discuss in detail in Sec.~\ref{sec:molmerformula}.

\textbf{General Case.} Now consider a more general parameter encoding. We see that $\alpha = \alpha_0 dt^2 + \alpha_1 dt$, $\gamma = \gamma_0 dt + \gamma_1 dt$. Here, the $\{\alpha_0,\gamma_0\}$ variables now encode the variance of the effective non-Hermitian Hamiltonian $H-\iota/2\sum_n L_n^\dagger L_n$. Whereas the $\{\alpha_1,\gamma_1\}$ parts contain the jump-only terms. Through a similar analysis as above, we see that in the continuous time limit only the effect of $\alpha_1$ survives, along with the $\beta_\tau$ terms.

\subsection{Mixed radiation state}\label{subsec:lowerboundonmixed}
In this section we formalize the notion that for a gapped dynamics, the majority of the QFI is contained exclusively in the radiated photons. That is, given the joint state $\ket{\psi}_{SR}$, if we define the mixed radiation state as $\rho_R := \tr_S[\ketbra{\psi}_{SR}]$, then $F[\rho_R] \approx F[\ket{\psi}_{SR}]$. We use the following result about the QFI of a mixed state to show this. 

\begin{theorem}[Ref.~\cite{kolodynski2013efficient}] \label{thm:dilatedqfi}
    The QFI of a mixed state $\rho$ on $R$ is equal to 
    \begin{equation}
        F[\rho_R] = \min_{\Psi} F[\ket{\Psi}_{SR}]
    \end{equation}
    where the minimization is over all purifications such that $\tr_S[\ketbra{\Psi}_{SR}] =  \rho_R$.
\end{theorem}

Using this theorem, we prove that the QFI of the mixed radiation state is asymptotically identical to \cref{eq:QFI_TIMPS}.
\begin{lem} \label{lem:mixedqfibound}
    The QFI of the mixed radiation state, for a normal MPS, is asymptotically identical to \cref{eq:QFI_TIMPS} up to a $O(\tau^* (\log{D})^k)$ correction for some $k$ depending on the Hamiltonian encoding. 
\end{lem}
\begin{proof}
    Let $\psi_{SR}$ denote the full state. Now let $G$ be any Hermitian matrix on $S$. By Theorem \ref{thm:dilatedqfi} we have that 
    \begin{equation}
        F[\rho_R] = \min_{G} F[U_\theta \ketbra{\psi_{SR}}U_\theta^\dagger]
    \end{equation}
    with $U_\theta := e^{-\iota G\theta}$. This is the most general form of rotations on the purifying register, encompassing all possible purifications. Now, we denote $\ket{\phi_{SR}} = U_\theta \ket{\psi_{SR}}$ and we have  
    \begin{align}
        \ket{\dot{\phi}}_{SR} = U_\theta \ket{\dot{\psi}}_{SR} - \iota G_S U_\theta \ket{\psi}_{SR} \\ 
        \bra{\dot{\phi}}_{SR} =  \bra{\dot{\psi}}_{SR}U_\theta^\dagger  + \iota   \bra{\psi}_{SR}U_\theta^\dagger G_S
    \end{align}
    Then, (emphasizing that $G$ is supported only on $S$)
    \begin{equation}
        F[\phi] = F[\psi] + \var_{\psi}[G_S] + 2\text{Im}\left(\expect{\dot{\psi}}{G_S}{\psi} -  \expect{\psi}{G_S}{\psi} \inner{\dot{\psi}}{\psi}\right)
    \end{equation}
    The only negative term in this expression is the cross-term, which we can show cannot be too large. 
    For large $T$, there exists a $|B\rrangle \approx (\sum_m K_m \otimes \dot{\bar{K}}_m)|\rho_{ss}\rrangle$ such that the cross-term satisfies, 
    \begin{equation}
        |\expect{\dot{\psi}}{G_S}{\psi} -  \expect{\psi}{G_S}{\psi} \inner{\dot{\psi}}{\psi}|= | \sum_\tau \llangle G |\tilde{T}_V^\tau| B \rrangle | \leq  \sum_\tau |  \llangle G |\tilde{T}_V^\tau| B \rrangle | \leq c \sum_{\tau}e^{-\tau/\tau^*} \leq  c \|G\| \frac{e^{-1/\tau^*}}{1-e^{-1/\tau^*}}(1-e^{-T/\tau^*})
    \end{equation} 
    where $c$ is a constant in time, scaling polynomially with particle number and $\tau^*$ is the dominant correlation length. Thus, we have 
    \begin{equation}
        F[\phi] \gtrsim F[\psi] + \var_{\psi}[G] - \|G\| c \tau^* 
    \end{equation}
    Note that $\|G\|$ has to be $O(1)$ in time in the minimization. To see this, first note that any $G$ which makes $\var_\psi[G]$ vanish must be uniform on the support of $\rho_S$. Hence, the covariance term vanishes identically, and the bound is trivial. Thus, without loss of generality we assume $\var[G] > 0$.  If $\|G\|$ scales with $T$, then $\var_{\psi}[G] - \|G\| c \tau^*  = \|G\|(c_1 \|G\| - c\tau^*)$ for some $c_1 > 0$ will become unbounded for large enough $T$. Hence we can take $\|G\| = O(1)$ without loss of generality for the minimization. Now applying Theorem~\ref{thm:dilatedqfi}
    \begin{equation}
        F[\rho_R] \gtrsim F[\psi_{SR}] - c\tau^*
    \end{equation}
    for some $c = \poly(N)$. The exact scaling of $c$ with $N$ can depend on the specific encoding.
\end{proof}

As a corollary, we have that for a gapped dynamics, 

\begin{equation}
    \lim_{T\to \infty} \frac{F_{SR}}{T} = \lim_{T\to \infty} \frac{F_{R}}{T}
\end{equation}

\section{Haar averaged QFI} \label{sec:jointQFI}
Now, we pay attention to the full output state upon $t$ iterations of the isometry $V$. Starting from the system state $\ket{\psi} \in \mathcal{H}_D$, we have the state at time $t$ denoted $\ket{\Psi;t} \in \C^{D\times d^t}$ as,
\begin{align}
    \ket{\Psi;t} = \quad
\begin{tikzpicture}[scale=0.5,baseline={([yshift=-0.65ex] current bounding box.center) }]
        \GTensor{0,0}{1.2}{.6}{\small $V$}{7};
      \GTensor{1*\singledx,0}{1.2}{.6}{\small $V$}{7};
      \GTensor{2*\singledx,0}{1.2}{.6}{\small $V$}{7};
      \GTensor{3*\singledx,0}{1.2}{.6}{\small $V$}{7};
      \GTensor{4*\singledx,0}{1.2}{.6}{\small $V$}{7};
    \DTensor{5*\singledx,0}{1.2}{.6}{\small ${\psi}$}{7};
\end{tikzpicture}
    \quad 
\end{align}
We recall that the initial state is not parameter-dependent. Thus, we have,
\begin{align}
\partial_\theta\ket{\Psi;t} = \mathlarger{\sum}_{\tau=1}^t    \quad\dots
\begin{tikzpicture}[scale=0.5,baseline={([yshift=-0.65ex] current bounding box.center) }]
        \GTensor{0,0}{1.2}{.6}{\small $V$}{7};
      \GTensor{1*\singledx,0}{1.2}{.6}{\small $V$}{7};
      \draw (1*\singledx,-1.8) node {\small $m_{\tau-1}$};
      \GTensor{2*\singledx,0}{1.2}{.6}{\small $\dot{V}$}{7};
      \draw (2*\singledx,-1.8) node {\small $m_{\tau}$};
      \GTensor{3*\singledx,0}{1.2}{.6}{\small $V$}{7};
      \draw (3*\singledx,-1.8) node {\small $m_{\tau+1}$};
      \GTensor{4*\singledx,0}{1.2}{.6}{\small $V$}{7}
      \DTensor{5*\singledx,0}{1.2}{.6}{\small ${\psi}$}{7};
\end{tikzpicture}
    \quad 
\end{align} 
The QFI is given as,
\begin{equation}
    F_{SR}(t;\psi,\theta,V) = 4\left(\inner{\partial_\theta\Psi;t}{\partial_\theta\Psi;t} - |\inner{\partial_\theta\Psi;t}{\Psi;t}|^2\right)
\end{equation}
We formally root out the initial state dependence by performing a Haar average over all possible initial states. 
\begin{equation}
  \langle F_{SR}\rangle (t) = \mathbb{E}_{\psi\sim\mu_N}[F_{SR}(t;\psi,\theta,V)] = 4\left( \mathbb{E}_{\psi\sim\mu_N}\inner{\partial_\theta\Psi;t}{\partial_\theta\Psi;t} -  \mathbb{E}_{\psi\sim\mu_N} |\inner{\partial_\theta\Psi;t}{\Psi;t}|^2\right)
\end{equation}
where $\mu_H$ denotes the Haar measure. 

\subsection{Unique fixed point}\label{subsec:haarunique}

We start again with the assumption that the dynamics has a unique steady state, and the gapped transfer matrix has a unique fixed point. Now, the first term is, 
\begin{align}
\underset{\psi \sim \mu_H}{\mathbb{E}}\quad
\begin{tikzpicture}[scale=0.5,baseline={([yshift=-0.65ex] current bounding box.center) }]
\GFTensor{0,2}{1.2}{.6}{\small $\partial_\theta V$}{1};
\GFTensor{0,0}{1.2}{.6}{\small $\partial_\theta \bar{V}$}{0};
\draw (1,1) node {\small \dots};
\draw[very thick,draw=red] (-1.2,0)--(-1.2,2);
\DTensor{3.6,0}{1.2}{.6}{\small $\bar{\psi}$}{7};
\DTensor{3.6,2}{1.2}{.6}{\small ${\psi}$}{7};
\end{tikzpicture}
  =   \frac{1}{D}\quad
\begin{tikzpicture}[scale=0.5,baseline={([yshift=-0.65ex] current bounding box.center) }]
\GFTensor{0,2}{1.2}{.6}{\small $\partial_\theta V$}{1};
\GFTensor{0,0}{1.2}{.6}{\small $\partial_\theta \bar{V}$}{0};
\draw (1,1) node {\small \dots};
\draw[very thick,draw=red] (-1.2,0)--(-1.2,2);
\draw[very thick,draw=red] (3,0)--(3,2);
\end{tikzpicture} \quad = \frac{1}{D} \left[ \tilde{\alpha}(t) + \sum_{\tau=0}^{t-2} \tilde{\beta}_\tau(t) \right]
\end{align} 

where

\begin{align}\label{eq:alphatwiddle}
  \tilde{\alpha}(t) := \quad
\begin{tikzpicture}[scale=0.5,baseline={([yshift=-0.65ex] current bounding box.center) }]
      \GTensor{0,0}{1.2}{.6}{\small $\dot{V}$}{7};
      \GTensor{0,-2}{1.2}{.6}{\small $\dot{\bar{V}}$}{0};
    \FUnitary{1.75,-2.25}{1}{.6}{\small $\sum T$}{6};
        \draw[very thick,draw=red] (2.355,0)--(3.5,0);
        \draw[very thick,draw=red] (2.355,-2)--(3.5,-2);
        \draw[very thick,draw=red] (3.5,-2)--(3.5,0);
        \draw[very thick,draw=red] (-2,0)--(-1.2,0);
        \draw[very thick,draw=red] (-2,-2)--(-1.2,-2);
        \draw[very thick,draw=red] (-2,-2)--(-2,0);
\end{tikzpicture}
    \quad 
\end{align} 
with $\sum T:= \sum_{\tau = 0}^{t-1} T_V^\tau$, and,

\begin{align}\label{eq:betastwiddle}
  \tilde{\beta}_\tau(t) := \quad
 \begin{tikzpicture}[scale=0.5,baseline={([yshift=-0.65ex] current bounding box.center) }]
       \GTensor{-5,0}{1.2}{.6}{\small $\dot{V}$}{7};
       \GTensor{-5,-2}{1.2}{.6}{\small $\bar{V}$}{0};
         \draw[very thick,draw=red] (-5-1.15,-2)--(-5-1.15,0);
         \FUnitary{-3.25,-2.25}{1}{.6}{\small $T_V^\tau$}{5};
       \GTensor{-1.45,0}{1.2}{.6}{\small ${V}$}{7};
       \GTensor{-1.45,-2}{1.2}{.6}{\small $\dot{\bar{V}}$}{0};
           \draw[very thick,draw=red] (-0.5,-2)--(0.15,-2);
           \draw[very thick,draw=red] (-0.5,0)--(0.15,0);
     \FUnitary{0.75,-2.25}{1}{.6}{\small $M_{\tau}$}{6};
         \draw[very thick,draw=red] (1.35,0)--(2,0);
         \draw[very thick,draw=red] (1.35,-2)--(2,-2);
         \draw[very thick,draw=red] (2,-2)--(2,0);
       \end{tikzpicture}
     \quad
 \end{align} 
 where $M_{\tau} := 2\sum_{\tau'=0}^{t-2-\tau}T_V^{\tau'}$. But, note that $\tilde{\alpha}(t)$ is \textit{exactly} equal $D\alpha t$, as all eigenmatrices of the Liouvillian with non-zero eigenvalue have trace zero (see Sec.~\ref{sec:statbasisprops} for details). Similarly, we have that $\tilde{\beta}_\tau(t) = 2D(t-\tau-1)\beta_\tau$. As a result, the first term is exactly $\alpha t + 2\sum_{\tau=0}^{t-2}(t-\tau-1)\beta_\tau$.

 Now, the second term is, 

\begin{align}
\underset{\psi \sim \mu_H}{\mathbb{E}}\quad
\begin{tikzpicture}[scale=0.5,baseline={([yshift=-0.65ex] current bounding box.center) }]
\GFTensor{0,2}{1.2}{.6}{\small $\partial_\theta V$}{1};
\GFTensor{0,0}{1.2}{.6}{\small $\bar{V}$}{0};
\draw (1,1) node {\small \dots};
\draw[very thick,draw=red] (-1.2,0)--(-1.2,2);
\DTensor{3.6,0}{1.2}{.6}{\small $\bar{\psi}$}{7};
\DTensor{3.6,2}{1.2}{.6}{\small ${\psi}$}{7};
\GFTensor{0,-2}{1.2}{.6}{\small $V$}{1};
\GFTensor{0,-4}{1.2}{.6}{\small $\partial_\theta\bar{V}$}{0};
\draw (1,-3) node {\small \dots};
\draw[very thick,draw=red] (-1.2,-4)--(-1.2,-2);
\DTensor{3.6,-4}{1.2}{.6}{\small $\bar{\psi}$}{7};
\DTensor{3.6,-2}{1.2}{.6}{\small ${\psi}$}{7};
\end{tikzpicture} \quad =   
&\frac{1}{D(D+1)}\left(\quad
\begin{tikzpicture}[scale=0.5,baseline={([yshift=-0.65ex] current bounding box.center) }]
\GFTensor{0,2}{1.2}{.6}{\small $\partial_\theta V$}{1};
\GFTensor{0,0}{1.2}{.6}{\small $\bar{V}$}{0};
\draw (1,1) node {\small \dots};
\draw[very thick,draw=red] (-1.2,0)--(-1.2,2);
\draw[very thick,draw=red] (3,0)--(3,2);
\GFTensor{0,-2}{1.2}{.6}{\small $V$}{1};
\GFTensor{0,-4}{1.2}{.6}{\small $\partial_\theta\bar{V}$}{0};
\draw (1,-3) node {\small \dots};
\draw[very thick,draw=red] (-1.2,-4)--(-1.2,-2);
\draw[very thick,draw=red] (3,-4)--(3,-2);
\end{tikzpicture}
  \quad  +  \quad 
\begin{tikzpicture}[scale=0.5,baseline={([yshift=-0.65ex] current bounding box.center) }]
\GFTensor{0,2}{1.2}{.6}{\small $\partial_\theta V$}{1};
\GFTensor{0,0}{1.2}{.6}{\small $\bar{V}$}{0};
\draw (1,1) node {\small \dots};
\draw[very thick,draw=red] (-1.2,0)--(-1.2,2);
\GFTensor{0,-2}{1.2}{.6}{\small $V$}{1};
\GFTensor{0,-4}{1.2}{.6}{\small $\partial_\theta\bar{V}$}{0};
\draw (1,-3) node {\small \dots};
\draw[very thick,draw=red] (-1.2,-4)--(-1.2,-2);
\draw[very thick,draw=red] (3,0)--(3,-2);
\draw[very thick,draw=red] (4,2)--(4,-4);
\draw[very thick,draw=red] (3,2)--(4,2);
\draw[very thick,draw=red] (3,-4)--(4,-4);
\end{tikzpicture} \quad\right) \quad = \frac{1}{D(D+1)}(|\tilde{\gamma}(t)|^2 + \tilde{\delta}(t))
\end{align} 

where we define 

\begin{align}\label{eq:gammatwiddle}
  \tilde{\gamma}(t) := \quad
\begin{tikzpicture}[scale=0.5,baseline={([yshift=-0.65ex] current bounding box.center) }]
      \GTensor{0,0}{1.2}{.6}{\small $\dot{V}$}{7};
      \GTensor{0,-2}{1.2}{.6}{\small $\bar{V}$}{0};
    \FUnitary{1.75,-2.25}{1}{.6}{\small $\sum T$}{6};
        \draw[very thick,draw=red] (2.355,0)--(3.5,0);
        \draw[very thick,draw=red] (2.355,-2)--(3.5,-2);
        \draw[very thick,draw=red] (3.5,-2)--(3.5,0);
        \draw[very thick,draw=red] (-2,0)--(-1.2,0);
        \draw[very thick,draw=red] (-2,-2)--(-1.2,-2);
        \draw[very thick,draw=red] (-2,-2)--(-2,0);
\end{tikzpicture}
    \quad 
\end{align}

\begin{align}
  \tilde{\delta}(t) := \quad
\begin{tikzpicture}[scale=0.5,baseline={([yshift=-0.65ex] current bounding box.center) }]
      \GTensor{0,0}{1.2}{.6}{\small $\dot{V}$}{7};
      \GTensor{0,-2}{1.2}{.6}{\small $\bar{V}$}{0};
    \FUnitary{1.75,-2.25}{1}{.6}{\small $\sum T$}{6};
    \draw[very thick,draw=red] (2.355,0)--(5,0);
    \draw[very thick,draw=red] (2.355,-2)--(3.5,-2);
    \draw[very thick,draw=red] (-2,0)--(-1.2,0);
    \draw[very thick,draw=red] (-2,-2)--(-1.2,-2);
    \draw[very thick,draw=red] (-2,-2)--(-2,0);
    \GTensor{0,-4}{1.2}{.6}{\small $V$}{7};
        \GTensor{0,-6}{1.2}{.6}{\small $\dot{\bar{V}}$}{0};
        \FUnitary{1.75,-6.25}{1}{.6}{\small $\sum T$}{6};
        \draw[very thick,draw=red] (2.355,-4)--(3.5,-4);
        \draw[very thick,draw=red] (2.355,-6)--(5,-6);
        \draw[very thick,draw=red] (-2,-4)--(-1.2,-4);
        \draw[very thick,draw=red] (-2,-6)--(-1.2,-6);
        \draw[very thick,draw=red] (-2,-6)--(-2,-4);
        \draw[very thick,draw=red] (3.5,-2)--(3.5,-4);
        \draw[very thick,draw=red] (5,-6)--(5,0);
          \end{tikzpicture}
    \quad 
\end{align} 

With the same reasoning as before, one has that $\tilde{\gamma}(t) = D\gamma t$ and $\tilde{\delta}(t) = D |\gamma|^2 t^2 + c_1$ with $c_1$ decaying exponentially in time. To see this, note that $\Sigma T = t|\rho_{ss}\rrangle \llangle\mathbbm{1}| + \Sigma \tilde{T}$, and the cross-terms between the two $\Sigma T$ cancel out owing to the fact that non-zero eigenmatrices are traceless. The $\Sigma \tilde{T} \times \Sigma \tilde{T}$ terms decays in time as $|\Lambda(\tilde{T}_V)| < 1$. Thus, the second term is $|\gamma|^2t^2 + c_1/D$ with $c_1 = O(||\mathbf{K}^\dagger \dot{\mathbf{K}}||_\infty^2 e^{-2\Delta t})$. Physical Hamiltonian encodings have norm $\poly(N) = \poly(\log{D})$. Thus, we have, 

\begin{result}[Haar averaged QFI]
  For a dynamics with a unique steady state,
  \begin{align}
    \underset{\psi \sim \mu_H}{\mathbb{E}} F_{SR}(t;\psi,\theta,V) &= 
      4 \left( \alpha t + 2\sum_{\tau=0}^{t-2}(t-\tau-1)\beta_\tau -  |\gamma|^2t^2 \right)  + c(t)
  \end{align}
  for $c(t) = O(e^{-2\Delta t}\poly(\log D))$.
\end{result}


An implication of this bound is that for an mixing system, entanglement in the initial state cannot act as a resource for Heisenberg limited sensing, either of the full state or in the case of noisy metrology.

\subsection{Generic case \label{subsec:haargeneric}}

Now, we turn to the non-normal case and provide the average value of the joint QFI. For purely imaginary eigenmatrices of the Liouvillian, note that the trace condition projects the variables onto the strictly-zero eigenspace of the Liouvillian (see Sec.~\ref{sec:statbasisprops} for details). We thus have,

\begin{align}
  \frac{\tilde{\alpha}(t)}{Dt}:= \sum_{\mu} \tr\left(J^\mu \frac{\mathbbm{1}}{D}\right) \quad
\begin{tikzpicture}[scale=0.5,baseline={([yshift=-0.65ex] current bounding box.center) }]
      \GTensor{0,0}{1.2}{.6}{\small $\dot{V}$}{7};
      \GTensor{0,-2}{1.2}{.6}{\small $\dot{\bar{V}}$}{0};
      \draw[very thick,draw=red] (1.2,0)--(2.03,0);
    \DTensor{2,-1}{1.}{.6}{\small $\Psi_{\mu}$}{1};
          \draw[very thick,draw=red] (1.2,-2)--(2.03,-2);
        \draw[very thick,draw=red] (-2,0)--(-1.2,0);
        \draw[very thick,draw=red] (-2,-2)--(-1.2,-2);
    \UTensor{-2,-1}{1.}{.6}{\small $J^{\mu}$}{1};
\end{tikzpicture}
    \quad \equiv \quad  \sum_{\mu} \tr\left(J^\mu \frac{\mathbbm{1}}{D}\right) \alpha^\mu_\mu 
\end{align} 

\begin{align}
  \frac{\tilde{\beta}_\tau(t)}{D}:= 2(t-\tau-1)\sum_{\mu} \tr\left(J^\mu \frac{\mathbbm{1}}{D}\right) \quad
\begin{tikzpicture}[scale=0.5,baseline={([yshift=-0.65ex] current bounding box.center) }]
      \GTensor{-5,0}{1.2}{.6}{\small $\dot{V}$}{7};
      \GTensor{-5,-2}{1.2}{.6}{\small $\bar{V}$}{0};
          \UTensor{-7,-1}{1.}{.6}{\small $J^{\mu}$}{1};
    \FUnitary{-3.25,-2.25}{1}{.6}{\small $T_V^\tau$}{5};
    \draw[very thick,draw=red] (-6,-2)--(-7,-2);
        \draw[very thick,draw=red] (-6,0)--(-7,0);
      \GTensor{-1.45,0}{1.2}{.6}{\small ${V}$}{7};
      \GTensor{-1.45,-2}{1.2}{.6}{\small $\dot{\bar{V}}$}{0};
          \DTensor{.5,-1}{1.}{.6}{\small $\Psi_{\mu}$}{1};
          \draw[very thick,draw=red] (-0.5,-2)--(0.53,-2);
          \draw[very thick,draw=red] (-0.5,0)--(0.53,0);
      \end{tikzpicture}
    \quad  \equiv \quad  2(t-\tau-1)\sum_{\mu} \tr\left(J^\mu \frac{\mathbbm{1}}{D}\right)\beta_\mu^\mu 
\end{align} 

\begin{align}
  \frac{\tilde{\gamma}(t)}{Dt} := \sum_{\mu} \tr\left(J^\mu \frac{\mathbbm{1}}{D}\right) \quad
\begin{tikzpicture}[scale=0.5,baseline={([yshift=-0.65ex] current bounding box.center) }]
      \GTensor{0,0}{1.2}{.6}{\small $\dot{V}$}{7};
      \GTensor{0,-2}{1.2}{.6}{\small $\bar{V}$}{0};
      \draw[very thick,draw=red] (1.2,0)--(2.03,0);
    \DTensor{2,-1}{1.}{.6}{\small $\Psi_{\mu}$}{1};
          \draw[very thick,draw=red] (1.2,-2)--(2.03,-2);
        \draw[very thick,draw=red] (-2,0)--(-1.2,0);
        \draw[very thick,draw=red] (-2,-2)--(-1.2,-2);
    \UTensor{-2,-1}{1.}{.6}{\small $J^{\mu}$}{1};
\end{tikzpicture}
    \quad  \equiv \quad  \sum_{\mu} \tr\left(J^\mu \frac{\mathbbm{1}}{D}\right) \gamma^\mu_\mu 
\end{align}

\begin{align}
  \frac{\tilde{\delta}(t)}{D^2 t^2}:= \sum_{\mu\nu} \tr\left(J^\mu J^\nu \right) \quad
\begin{tikzpicture}[scale=0.5,baseline={([yshift=-0.65ex] current bounding box.center) }]
      \GTensor{0,0}{1.2}{.6}{\small $\dot{V}$}{7};
      \GTensor{0,-2}{1.2}{.6}{\small $\bar{V}$}{0};
      \draw[very thick,draw=red] (1.2,0)--(2.03,0);
    \DTensor{2,-1}{1.}{.6}{\small $\Psi_{\mu}$}{1};
          \draw[very thick,draw=red] (1.2,-2)--(2.03,-2);
        \draw[very thick,draw=red] (-2,0)--(-1.2,0);
        \draw[very thick,draw=red] (-2,-2)--(-1.2,-2);
    \UTensor{-2,-1}{1.}{.6}{\small $\mathbbm{1}/D$}{1};
    \begin{scope}[shift={(6cm,0cm)}, scale=1.]
      \GTensor{0,0}{1.2}{.6}{\small $\dot{\bar{V}}$}{7};
      \GTensor{0,-2}{1.2}{.6}{\small $\bar{V}$}{0};
      \draw[very thick,draw=red] (1.2,0)--(2.03,0);
    \DTensor{2,-1}{1.}{.6}{\small $\Psi_{\nu}$}{1};
          \draw[very thick,draw=red] (1.2,-2)--(2.03,-2);
        \draw[very thick,draw=red] (-2,0)--(-1.2,0);
        \draw[very thick,draw=red] (-2,-2)--(-1.2,-2);
    \UTensor{-2,-1}{1.}{.6}{\small $\mathbbm{1}/D$}{1};
  \end{scope}    
\end{tikzpicture}
    \quad  + O(e^{-\Delta t}) \equiv  \quad \sum_{\mu\nu} \tr\left(J^\mu J^\nu \right) \gamma_\mu \bar{\gamma}_\nu + O(e^{-\Delta t})
\end{align}

Then, we have that, 

\begin{result}[Haar average for degenerate MPS]
  \begin{align}
    \underset{\psi \sim \mu_H}{\mathbb{E}} F_{SR}(t;\psi,\theta,V) = &\sum_\mu \tr\left(J^\mu \frac{\mathbbm{1}}{D}\right) \left( t\alpha^\mu_\mu + 2\sum_\tau (t-\tau-1) \beta_\mu^\mu \right) \\ & -  t^2 \frac{D}{D+1} \left[ \sum_{\mu\nu} \tr\left(J^\mu \frac{\mathbbm{1}}{D}\right) \tr\left(J^\nu \frac{\mathbbm{1}}{D}\right) \gamma_{\mu}^\mu \bar{\gamma}_\nu^\nu + \tr(J^\mu J^\nu) \gamma_\mu \bar{\gamma}_\nu \right] + c'(t)
  \end{align}     
  for $c'(t) = O(e^{-2\Delta t}\poly(\log D))$.
\end{result}



\section{Properties of the steadystate spectrum} \label{sec:statbasisprops}
We outline some useful properties of the (left/right) eigenspectrum of a Liouvillian superoperator. We denote the vectorization of $\mathcal{L}(O)$as $\mathcal{L}|O\rrangle$, and other similar notations. In general, we have the following form of the Liouvillian,
\begin{equation}
    \mathcal{L} = \sum_{\lambda,\mu}\lambda |\Psi_{\lambda\mu}\rrangle \llangle J^{\lambda\mu}|
\end{equation}
where $\lambda$ denotes the eigenvalue with $\text{Re}(\lambda) \leq 0$ and $\mu$ (possibly depending on $\lambda$) denotes a degeneracy factor. In the main text we are only concerned with the stationary basis with $\lambda = \iota\Delta$. Moreover, we have biorthogonality of this linear system, viz., $\llangle J^{\lambda\mu}|\Psi_{\lambda'\nu}\rrangle = \delta^{\lambda}_{\lambda'}\delta^\mu_{\nu}$. For any operator $O \in L(\mathcal{H})$, we have that
\begin{align}\label{eq:daggerclosure}
    \mathcal{L}(O) = 0 &\Leftrightarrow \mathcal{L}(O^\dagger) = 0 \\ 
    \mathcal{L}^\dagger(O) = 0 &\Leftrightarrow \mathcal{L}^\dagger(O^\dagger) = 0
\end{align}
Now, let $\{J^\mu\}$ be a basis of the left eigenmatrices with eigenvalue zero of $\mathcal{L}$. We can convert this to a Hermitian basis, with $J^\mu = J^{\mu\dagger}$ using \eqref{eq:daggerclosure}. Furthermore, any Hermitian operator has a decomposition $H = A - B$ with $A, B \geq 0$ being positive-semidefinite. With this, we now settle on a basis $\{J^\mu\}$ which is PSD, and further normalize $\text{Tr}(J^\mu) = 1 ~\forall \mu$. Now, since the Jordan block is diagonal, there exist $\Psi_\nu$ the corresponding right eigenmatrices which are biorthogonal, viz., 
\begin{eqnarray}
    \llangle J^\mu|\Psi_\nu\rrangle = \delta^\mu_\nu
\end{eqnarray}
Moreover, the right eigenmatrices can be chosen to be Hermitian. This in-turn ensures that $c_\mu := \llangle J^\mu|\rho_{in}\rrangle \geq 0$ as both $J^\mu$ and $\rho_{in}$ are PSD. Furthermore, the PSD property implies that 
\begin{eqnarray}
    c_\mu = \llangle J^\mu|\rho_{in}\rrangle = \text{Tr}(J^\mu\rho_{in}) \leq \text{Tr}(J^\mu)\text{Tr}(\rho_{in}) = 1
\end{eqnarray}
Summarising, we have $c_\mu \in [0,1]$. Note that while this degree of freedom and choice of the left eigenmatrices is well defined, this in-turn fixes the trace of the corresponding right-eigenvectors via the biorthogonality condition, and they are not trace one in general. Nonetheless, the right traces do not enter the analysis explicity.

Now, we note some properties for the $\Delta \neq 0$ purely imaginary eigenvalues of the Liouvillian. Note that, for any $O \in L(\mathcal{H})$, 
\begin{align}
    \mathcal{L}(O) = \iota\Delta O &\Leftrightarrow \mathcal{L}(O^\dagger) = -\iota\Delta O^\dagger \\ 
    \mathcal{L}^\dagger(O) = \iota\Delta O &\Leftrightarrow \mathcal{L}^\dagger(O^\dagger) = -\iota\Delta O^\dagger 
\end{align}
Thus, the eigenmatrices in the $\Delta \neq 0$ subspace cannot be made Hermitian. The physical density matrices consisting of such operators are indeed $\Psi_{\Delta\mu} + \Psi_{\Delta\mu}^\dagger$ and $\iota(\Psi_{\Delta\mu} - \Psi_{\Delta\mu}^\dagger)$. Additionally, we have that
\begin{align}
    &c_{\Delta\mu} := \llangle J^{\Delta\mu}|\rho_{in}\rrangle = \text{Tr}(J^{\Delta\mu\dagger}\rho_{in}) \\ 
    &\implies  c_{\Delta\mu}^* = \text{Tr}(\rho_{in}J^{\Delta\mu}) = c_{-\Delta\mu}
\end{align}
Also, any non-zero (left or right) eigenoperators are traceless, viz.,
\begin{eqnarray}
    \mathcal{L}^{(\dagger)}(O) = \lambda O,~\lambda \neq 0 \implies \text{Tr}(O) = 0
\end{eqnarray}
This fact is used for analyzing cross-terms in the average joint QFI of Sec.~\ref{sec:jointQFI}.
\section{Relation to the formula in previous works}\label{sec:molmerformula}
A numerical method for computing the total QFI of the joint state of the system and the radiation was given in \cite{molmer_14}. We briefly recap it here, and compare the asymptotic rate to that derived in Sec.~\ref{sec:qfinormalproof} from the MPS calculations. For two values of the parameter $\theta_1,\theta_2\in\Theta$, we define a `tilted' Liouivillian, given in the vectorized form as,
\begin{align}
    &\overset{\leftrightarrow}{\mathcal{L}}_{\theta_1 \theta_2} = -\iota H(\theta_1) \otimes \mathbbm{1} + \iota \mathbbm{1} \otimes H(\theta_2)^\top \\
    &\quad + \sum_{n \geq 1} \left[ L_n(\theta_1) \otimes \bar{L}_n(\theta_2)  - \frac{1}{2} \left( L_n^{\dagger}(\theta_1) L_n(\theta_1) \otimes \mathbbm{1} + \mathbbm{1} \otimes L_n^\top(\theta_2) \bar{L}_n(\theta_2) \right) \right]
\end{align}
where we assume that the parameter is encoded in the Hamiltonian as well as the jump operator. Then, one evolves the initial state $|\rho_0\rrangle$ under this tilted dynamics to get $|\rho^{\theta_1\theta_2}(t)\rrangle = \exp{\overset{\leftrightarrow}{\mathcal{L}}_{\theta_1 \theta_2} t}|\rho_0\rrangle$, and computes the total QFI as \cite{molmer_14},
\begin{equation} \label{eq:MolmerFormula}
    F(t;\theta) = 4\partial_{\theta_1}\partial_{\theta_2}\log{|\llangle\mathbbm{1}|\rho^{\theta_1\theta_2}(t)\rrangle|}~\mid_{\theta_1=\theta_2=\theta}
\end{equation}
The key idea behind this is that the QFI of the total state $\Psi$ can be written as
\begin{align}
4\partial_{\theta_1}\partial_{\theta_2}\log|\inner{\Psi(\theta_1)}{\Psi(\theta_2)}|\mid_{\theta_1=\theta_2=\theta} =  4\partial_{\theta_1}\partial_{\theta_2}\log|\tr\ket{\Psi(\theta_2)}\bra{\Psi(\theta_1)}|\mid_{\theta_1=\theta_2=\theta}
\end{align}
and the trace over the emission field (environment) can be taken inside, leading to the $\theta_1$ evolution on the left-multiplied operators, and $\theta_2$ evolution on the right-multiplied operators. Finally, the system trace, represented in \cref{eq:MolmerFormula} as the inner product with $|\mathbbm{1}\rrangle$, leads to the desired result.

Now, the following ansatz is made for long times $t$ \cite{molmer_14},
\begin{equation}
    |\llangle\mathbbm{1}|\rho^{\theta_1\theta_2}(t)\rrangle| \sim e^{t \lambda(\theta_1,\theta_2)} c(\theta_1,\theta_2)
\end{equation}
where $\lambda$ is the leading eigenvalue of the tilted generator, which obeys 
\begin{equation}
    \lim_{\theta_1,\theta_2 \to \theta} \lambda(\theta_1,\theta_2) = 0
\end{equation}
and $c(\theta_1,\theta_2)$ is an expansion coefficient of the initial state onto the leading eigenmatrix. Assuming this ansatz is true, one gets that the QFI is $4t \partial_{\theta_1}\partial_{\theta_2} \lambda(\theta_1,\theta_2) +  O(1)$. . Now, a perturbation theory expansion of this derivative was derived in Supplemental Sec.~I~(A-B) of Ref.~\cite{molmer_14}. In terms of the variables defined in the main text, this expansion for $\partial_{\theta_1}\partial_{\theta_2} \lambda(\theta_1,\theta_2)$ reads, 
\begin{equation}
    \partial_{\theta_1}\partial_{\theta_2} \lambda(\theta_1,\theta_2) = \tr\left[\sum_{n > 0} \dot{L}_n \rho_{ss}\dot{L}_n^\dagger \right] - \tr\left[\mathcal{L}_L \circ \mathcal{L}^{-1}_{n >0} \circ \mathcal{L}_R \right]  
\end{equation}
with 
\begin{align}
    \mathcal{L}_L &:= \left[-i\dot{H} - \frac{1}{2} \sum_{n>0} \partial_\theta\left(L_n^{\dagger} L_n\right)\right] \rho_{ss} + \sum_{n>0} \dot{L}_n \rho_{ss} L_n^{\dagger} \\
    \mathcal{L}_R &:= \rho_{ss} \left[\iota \dot{H} - \frac{1}{2} \sum_{n>0} \partial_\theta\left(L_n^{\dagger} L_n\right)\right] + \sum_{n>0} L_n \rho_{ss} \dot{L}_n^{\dagger}
\end{align}
and $\mathcal{L}^{-1}_{n >0}$ is the Liouvillian pseudoinverse in the decaying subspace.
This is the same as the approximation made in Sec.~\ref{sec:qfinormalproof}
\begin{equation}
    (\alpha-|\gamma|^2) + 2\left(\sum_{\mu > 0} K_\mu \tau_\mu\right)
\end{equation}
in \cref{eq:asymptoticrate1}, upon taking the continuous time limit $dt \to 0$  where we get the coefficients $K_\mu$ as defined in Sec.~\ref{sec:qfinormalproof}. This verifies that the two formulations are asymptotic consistent. Moreover, the MPS calculations can be seen as a trotterized version of \cref{eq:MolmerFormula}, leading to a much more transparent result for the QFI as a function of time.
\section{Relation to noisy metrology} \label{sec:normboundssec}
Given a parametrized channel $\mathcal{E}$ described by the Kraus operators $\mathbf{K} = [K_0,K_1,\dots,K_{d-1}]$, we define the effective Hamiltonian as $H_s = \iota \mathbf{K}^\dagger \dot{\mathbf{K}}$. Further, we define the \textit{Kraus span} as $\mathsf{K}:= \text{span}(\{K_i^\dagger K_j \mid 0 \leq i,j < d\})$. This is a generalization of the Hamiltonian-in-Lindblad span condition outlined in the main text, and subsumes the \(H_s \in \mathsf{L} = \text{span}\left\{\mathbbm{1}, L_m, L_m^\dagger, L_m^\dagger L_n \right\}_{m,n}\). We say that the Hamiltonian is in the Kraus span (HKS) iff $H_s \in \mathsf{K}$. If not, we abbreviate it as HKNS, i.e., $H_s = H_\perp + H_{||}$ with $H_{\perp} \neq 0$ being orthogonal to $\mathsf{K}$.

\begin{theorem}[Theorem 1. of Ref.~\cite{zhou2021asymptotic}]
    Noisy QFI $F_S(t) = \Theta(t^2)$ iff the HKNS condition holds. Otherwise, $F_S(t) = \Theta(t)$.
\end{theorem}
Typically, one needs to actively perform error correction to achieve the quadratic scaling when HKNS is true. The intuition is that one can project onto the $H_\perp$ part by decoupling the radiation quanta from the system at each time step. Nonetheless, we will be more interested in the case when HKS is true. The key idea in proving this is to use monotonicity. We have the following,
\begin{align}
   F_{S} \leq  F_{SR}(\psi) \leq \quad
     \begin{tikzpicture}[scale=0.5,baseline={([yshift=-0.65ex] current bounding box.center) }]
        \GFTensor{0,2}{1.2}{.6}{\small $\dot{V}$}{1};
        \GFTensor{0,0}{1.2}{.6}{\small $\dot{\bar{V}}$}{0};
        \draw (1,1) node {\small \dots};
        \draw[very thick,draw=red] (-1.2,0)--(-1.2,2);
        \DTensor{3.6,0}{1.2}{.6}{\small $\bar{\psi}$}{7};
        \DTensor{3.6,2}{1.2}{.6}{\small ${\psi}$}{7};
    \end{tikzpicture}
    \quad \leq
    \quad
   \left\| 
    \begin{tikzpicture}[scale=0.5,baseline={([yshift=-0.65ex] current bounding box.center) }]
        \GFTensor{0,2}{1.2}{.6}{\small $\dot{V}$}{1};
        \GFTensor{0,0}{1.2}{.6}{\small $\dot{\bar{V}}$}{0};
        \draw (1,1) node {\small \dots};
        \draw[very thick,draw=red] (-1.2,0)--(-1.2,2);
    \end{tikzpicture}
   \right\|
   \quad 
\end{align}

An upper bound on the QFI of this state can be expressed in terms of the single-channel Kraus operators if one defines $A := \sum_{m}\dot{K}^\dagger_m\dot{K}_m = \dot{\mathbf{K}}^\dagger\dot{\mathbf{K}}$ and $B := \sum_m\dot{K}^\dagger_m K_m = \dot{\mathbf{K}}^\dagger{\mathbf{K}}$, then one has \cite{zhou2021asymptotic}
\begin{eqnarray}
    F_S \leq 4 \min_{\{h\}}\left(T\|A\| + (T^2-T)\|B\|(\|B\| + 2\sqrt{\|A\|})\right)
\end{eqnarray}
The quadratic part of the upper bound depends on $\|B\|$, which is easily seen by noting the following tensor representations of $A$ and $B$, 
\begin{align}\label{eq:Amatrix}
    |A\rrangle \equiv \quad
  \begin{tikzpicture}[scale=0.5,baseline={([yshift=-0.65ex] current bounding box.center) }]
        \GTensor{0,0}{1.2}{.6}{\small $\dot{V}$}{7};
        \GTensor{0,-2}{1.2}{.6}{\small $\dot{\bar{V}}$}{0};
        \draw[very thick,draw=red] (1.2,0)--(2.03,0);
            \draw[very thick,draw=red] (1.2,-2)--(2.03,-2);
          \draw[very thick,draw=red] (-2,0)--(-1.2,0);
          \draw[very thick,draw=red] (-2,-2)--(-1.2,-2);
          \draw[very thick,draw=red] (-2,-2)--(-2,0);
  \end{tikzpicture}
      \quad 
  \end{align} 

  \begin{align}\label{eq:Bmatrix}
    |B\rrangle \equiv \quad
  \begin{tikzpicture}[scale=0.5,baseline={([yshift=-0.65ex] current bounding box.center) }]
        \GTensor{0,0}{1.2}{.6}{\small $V$}{7};
        \GTensor{0,-2}{1.2}{.6}{\small $\dot{\bar{V}}$}{0};
        \draw[very thick,draw=red] (1.2,0)--(2.03,0);
            \draw[very thick,draw=red] (1.2,-2)--(2.03,-2);
          \draw[very thick,draw=red] (-2,0)--(-1.2,0);
          \draw[very thick,draw=red] (-2,-2)--(-1.2,-2);
          \draw[very thick,draw=red] (-2,-2)--(-2,0);
  \end{tikzpicture}
      \quad 
  \end{align} 
Specifically, it has been shown that HKS implies the existence of a purification $SR$ such that $B = 0$, in turn implying a linear bound on the noisy QFI. Physically, the information encoded about the parameter $\theta$ in $V_\theta$ `leaks out' into the environment, then it cannot be recovered optimally via any measurement on the system. We now formalize this assuming HKS is true. 

\subsection{Continuous sensing under HKS condition}\label{subsec:continuousHKS}

Note that HKS can be restated as follows, 
\begin{align} \label{eq:HKS}
    \quad 
  \begin{tikzpicture}[scale=0.5,baseline={([yshift=-0.65ex] current bounding box.center) }]
    \GTensor{0,0}{1.2}{.6}{\small $V$}{7};
    \GTensor{0,-2}{1.2}{.6}{\small $\dot{\bar{V}}$}{0};
    \draw[very thick,draw=red] (-1.2,-2)--(-1.2,0);
\end{tikzpicture} \quad \in \text{span} \{|K_{m}^{\dagger}K_{m'}\rrangle \}_{(m,m')\in[d]^2}
\end{align}

 \begin{lem} \label{lem:HKSImplies}
    If the isometry $V$ satisfies the HKS condition of \cref{eq:HKS} then we have that 
    \begin{align} \label{eq:VofHKS}
        \quad
        \begin{tikzpicture}[scale=0.5,baseline={([yshift=-0.65ex] current bounding box.center) }]
            \GTensor{0,0}{1.2}{.6}{\small $V_\theta$}{7};
      \end{tikzpicture} \quad \equiv 
     \quad
         \begin{tikzpicture}[scale=0.5,baseline={([yshift=-0.65ex] current bounding box.center) }]
             \GTensor{0,0}{1.2}{.6}{\small $W_\theta$}{7};
               \RTensor{0,-2}{1.}{.6}{\small $u_\theta$}{1};
       \end{tikzpicture}
           \quad 
       \end{align} 
       with $u = e^{\iota\theta h}$  and some Hermitian $h\in \mathbb{C}^{d\times d}$. That is, $\ket{\psi_V} = U_\theta \ket{\psi_W}$ with $U_\theta = \exp{\iota\theta \left(\sum_\tau h_\tau\right)}$. Furthermore, $W$ satisfies, 
       \begin{align} \label{eq:Wcondition}
        \begin{tikzpicture}[scale=0.5,baseline={([yshift=-0.65ex] current bounding box.center) }]
            \GTensor{0,0}{1.2}{.6}{\small $\dot{W}$}{7};
            \GTensor{0,-2}{1.2}{.6}{\small ${\bar{W}}$}{0};
            \draw[very thick,draw=red] (-1.2,-2)--(-1.2,0);
        \end{tikzpicture} \quad 
        = \quad 0 
    \end{align}
    As a consequence, $\beta^W_\tau = 0~\forall \tau$ and $\gamma^W = 0$, leading to $F[\psi_W] = 4T\cdot\alpha^W$.
 \end{lem}
This formalizes the intuition of information `leaking out' into the environment, through the unitary $u_\theta$. An extreme example of this would be the case where \textit{all} the information about the parameter flows into the photon space, and we have $\dot{W} \equiv 0$. In general, the Lemma tells us that $\ket{\psi_V} = e^{\iota \theta H}\ket{\psi_W}$ with $H = \sum_{\tau}h^{(\tau)}$. For noisy metrology, both $V$ and $W$ are valid purifications, and thus \cref{eq:Wcondition} implies a linear bound on the noisy QFI $F_S$ by monotonicity.

\begin{proof}
    Assume the Hamiltonian is in the Kraus span. Thus, there exists a Hermitian operator $h$ such that 
    \begin{align} 
        \begin{tikzpicture}[scale=0.5,baseline={([yshift=-0.65ex] current bounding box.center) }]
            \GTensor{0,0}{1.2}{.6}{\small $V$}{7};
            \GTensor{0,-2}{1.2}{.6}{\small $\dot{\bar{V}}$}{0};
            \draw[very thick,draw=red] (-1.2,-2)--(-1.2,0);
        \end{tikzpicture} \quad 
        = -\iota \quad 
        \begin{tikzpicture}[scale=0.5,baseline={([yshift=-0.65ex] current bounding box.center) }]
            \GTensor{0,0}{1.2}{.6}{\small $V$}{7};
            \GTensor{0,-4}{1.2}{.6}{\small $\bar{V}$}{0};
            \RTensor{0,-2}{1.}{.6}{\small $h$}{1};
            \draw[very thick,draw=red] (-1.2,-4)--(-1.2,0);
        \end{tikzpicture}
    \end{align}
Given $h$, define a new tensor $W_\theta$ as follows,
  \begin{align} \label{eq:defineW}
    \quad
    \begin{tikzpicture}[scale=0.5,baseline={([yshift=-0.65ex] current bounding box.center) }]
        \GTensor{0,0}{1.2}{.6}{\small $W_\theta$}{7};
  \end{tikzpicture} \quad := 
 \quad
     \begin{tikzpicture}[scale=0.5,baseline={([yshift=-0.65ex] current bounding box.center) }]
         \GTensor{0,0}{1.2}{.6}{\small $V_\theta$}{7};
           \RTensor{0,-2}{1.}{.6}{\small $u_\theta^\dagger$}{1};
   \end{tikzpicture}
       \quad 
   \end{align} 
Then note that 

\begin{align} 
    \quad
    \begin{tikzpicture}[scale=0.5,baseline={([yshift=-0.65ex] current bounding box.center) }]
        \GTensor{0,0}{1.2}{.6}{ $\dot{W}$}{7};
  \end{tikzpicture} \quad = 
 \quad
     \begin{tikzpicture}[scale=0.5,baseline={([yshift=-0.65ex] current bounding box.center) }]
         \GTensor{0,0}{1.2}{.6}{$\dot{V}$}{7};
           \RTensor{0,-2}{1.}{.6}{\small $u_\theta^\dagger$}{1};
   \end{tikzpicture}
       \quad - \iota \quad
\begin{tikzpicture}[scale=0.5,baseline={([yshift=-0.65ex] current bounding box.center) }]
    \GTensor{0,0}{1.2}{.6}{\small $V$}{7};
      \RTensor{0,-2}{1.}{.6}{\small $u_\theta^\dagger$}{1};
      \RTensor{0,-4}{1.}{.6}{\small $h$}{1};
\end{tikzpicture}
  \quad 
\end{align} 
And thus, 

\begin{align}
    \begin{tikzpicture}[scale=0.5,baseline={([yshift=-0.65ex] current bounding box.center) }]
        \GTensor{0,0}{1.2}{.6}{\small $W$}{7};
        \GTensor{0,-2}{1.2}{.6}{$\dot{\bar{W}}$}{0};
  \end{tikzpicture} \quad 
  = \quad 
  \begin{tikzpicture}[scale=0.5,baseline={([yshift=-0.65ex] current bounding box.center) }]
    \GTensor{0,0}{1.2}{.6}{\small $V$}{7};
    \GTensor{0,-2}{1.2}{.6}{ $\dot{\bar{V}}$}{0};
\end{tikzpicture} \quad 
+ \iota \quad 
\begin{tikzpicture}[scale=0.5,baseline={([yshift=-0.65ex] current bounding box.center) }]
    \GTensor{0,0}{1.2}{.6}{\small $V$}{7};
    \GTensor{0,-4}{1.2}{.6}{\small $\bar{V}$}{0};
    \RTensor{0,-2}{1.}{.6}{\small $h$}{1};
\end{tikzpicture}
\end{align}
leading to \cref{eq:Wcondition}. Inverting the photon-space unitary in \cref{eq:defineW}, we arrive at \cref{eq:VofHKS}. 
\end{proof}

Note that if one Kraus representation obeys HKS, then so do all others. However, the explicit matrix $h$ that one gets in the Lemma above may be different for the different representations. Nonetheless, we obtain a major simplification of the radiation QFI when HKS is true.

\begin{lem}\label{lem:HKSQFI}
    If the isometry $V$ obeys HKS, then the QFI $F[\psi_V]$ takes the following form,
    \begin{align}
        F[\psi_V] = 4\var_W[H] + 4T\alpha^W + 8\text{Im}\expect{\psi_W}{H}{\dot{\psi}_W}
    \end{align}
    with the effective generator $H = \sum_\tau h_\tau$.
\end{lem}  
\begin{proof}
    Now, note that (denoting $U_\theta := e^{\iota \theta H}$)
\begin{align}
    \ket{\dot{\psi}_V} &= \iota HU_\theta \ket{\psi_W} + U_\theta\ket{\dot{\psi}_W} \\ 
    \bra{\dot{\psi}_V} &= -\iota  \bra{\psi_W} U_\theta^\dagger H+ \bra{\dot{\psi}_W}U_\theta^\dagger
\end{align}
Thus, 
\begin{align}
    \inner{\dot{\psi}_V}{\dot{\psi}_V} &= \expect{\psi_W}{H^2}{\psi_W} + \inner{\dot{\psi}_W}{\dot{\psi}_W} + 2\text{Im}\expect{\psi_W}{H}{\dot{\psi}_W} \\ 
    \inner{\dot{\psi}_V}{{\psi}_V}  &= -\iota \expect{\psi_W}{H}{\psi_W} + \inner{\dot{\psi}_W}{\psi_W}
\end{align}
Now, since the $W$ tensor satisfies \cref{eq:Wcondition}, we have have that $\inner{\dot{\psi}_W}{\psi_W} = 0$. Furthermore, 
\begin{align} 
    \inner{\dot{\psi}_W}{\dot{\psi}_W}  \equiv T\cdot \quad
  \begin{tikzpicture}[scale=0.5,baseline={([yshift=-0.65ex] current bounding box.center) }]
        \GTensor{0,0}{1.2}{.6}{\small $\dot{W}$}{7};
        \GTensor{0,-2}{1.2}{.6}{\small $\dot{\bar{W}}$}{0};
        \draw[very thick,draw=red] (1.2,0)--(2.03,0);
      \DTensor{2,-1}{1.}{.6}{\small $\rho_{ss}$}{1};
        \draw[very thick,draw=red] (1.2,-2)--(2.03,-2);
        \draw[very thick,draw=red] (-2,0)--(-1.2,0);
        \draw[very thick,draw=red] (-2,-2)--(-1.2,-2);
        \draw[very thick,draw=red] (-2,-2)--(-2,0);
  \end{tikzpicture}
      \quad  \equiv T \cdot \alpha^W 
  \end{align} 
  as all the `non-diagonal' derivatives equate to zero as a consequence of \cref{eq:Wcondition}.  Adding up the terms, we get the result.
\end{proof}

Notice that for large $T$, 
\begin{align}
    \expect{\psi_W}{H}{\dot{\psi}_W} \simeq T \cdot \quad 
    \begin{tikzpicture}[scale=0.5,baseline={([yshift=-0.65ex] current bounding box.center) }]
        \GTensor{0,0}{1.2}{.6}{\small $\dot{W}$}{7};
        \GTensor{0,-4}{1.2}{.6}{\small $\bar{W}$}{0};
        \RTensor{0,-2}{1.}{.6}{\small $h$}{1};
        \draw[very thick,draw=red] (-1.2,-4)--(-1.2,0);
        \DTensor{2,-2}{1.}{.6}{\small $\rho_{ss}$}{1};
        \draw[very thick,draw=red] (2,-3)--(2,-4);
        \draw[very thick,draw=red] (2,-1)--(2,0);
        \draw[very thick,draw=red] (1,0)--(2,0);
        \draw[very thick,draw=red] (1,-4)--(2,-4);
    \end{tikzpicture}
    \quad 
\end{align}

\begin{result}
    If the dynamics generated by $V$ obeys HKS, the variables $\alpha,\gamma,\beta$ are noted to be 
    \begin{align}
        \alpha^V= \alpha^W + \quad 
        2\text{Im}\left(\begin{tikzpicture}[scale=0.5,baseline={([yshift=-0.65ex] current bounding box.center) }]
            \GTensor{0,0}{1.2}{.6}{\small $\dot{W}$}{7};
            \GTensor{0,-4}{1.2}{.6}{\small $\bar{W}$}{0};
            \RTensor{0,-2}{1.}{.6}{\small $h$}{1};
            \draw[very thick,draw=red] (-1.2,-4)--(-1.2,0);
            \DTensor{2,-2}{1.}{.6}{\small $\rho_{ss}$}{1};
            \draw[very thick,draw=red] (2,-3)--(2,-4);
            \draw[very thick,draw=red] (2,-1)--(2,0);
            \draw[very thick,draw=red] (1,0)--(2,0);
            \draw[very thick,draw=red] (1,-4)--(2,-4);
        \end{tikzpicture}\right) \quad  + \quad 
        \begin{tikzpicture}[scale=0.5,baseline={([yshift=-0.65ex] current bounding box.center) }]
            \GTensor{0,0}{1.2}{.6}{\small ${W}$}{7};
            \GTensor{0,-4}{1.2}{.6}{\small $\bar{W}$}{0};
            \RTensor{0,-2}{1.}{.6}{\small $h^2$}{1};
            \draw[very thick,draw=red] (-1.2,-4)--(-1.2,0);
            \DTensor{2,-2}{1.}{.6}{\small $\rho_{ss}$}{1};
            \draw[very thick,draw=red] (2,-3)--(2,-4);
            \draw[very thick,draw=red] (2,-1)--(2,0);
            \draw[very thick,draw=red] (1,0)--(2,0);
            \draw[very thick,draw=red] (1,-4)--(2,-4);
        \end{tikzpicture} \quad 
    \end{align}
    
    \begin{align}
        \gamma^V = \quad 
        \iota \cdot \quad \begin{tikzpicture}[scale=0.5,baseline={([yshift=-0.65ex] current bounding box.center) }]
            \GTensor{0,0}{1.2}{.6}{\small ${W}$}{7};
            \GTensor{0,-4}{1.2}{.6}{\small $\bar{W}$}{0};
            \RTensor{0,-2}{1.}{.6}{\small $h$}{1};
            \draw[very thick,draw=red] (-1.2,-4)--(-1.2,0);
            \DTensor{2,-2}{1.}{.6}{\small $\rho_{ss}$}{1};
            \draw[very thick,draw=red] (2,-3)--(2,-4);
            \draw[very thick,draw=red] (2,-1)--(2,0);
            \draw[very thick,draw=red] (1,0)--(2,0);
            \draw[very thick,draw=red] (1,-4)--(2,-4);
        \end{tikzpicture}\quad 
    \end{align}

    \begin{align}
        \beta^V_\tau  =  \quad \begin{tikzpicture}[scale=0.5,baseline={([yshift=-0.65ex] current bounding box.center) }]
            \GTensor{0,0}{1.2}{.6}{\small ${W}$}{7};
            \GTensor{0,-4}{1.2}{.6}{\small $\bar{W}$}{0};
            \RTensor{0,-2}{1.}{.6}{\small $h$}{1};
            \draw[very thick,draw=red] (-1.2,-4)--(-1.2,0);
            \FUnitary{2.,-4}{1}{1.}{\small $T_V^\tau$}{5};
            \GTensor{4.25,0}{1.2}{.6}{\small ${W}$}{7};
            \GTensor{4.25,-4}{1.2}{.6}{\small $\bar{W}$}{0};
            \RTensor{4.25,-2}{1.}{.6}{\small $h$}{1};
            \DTensor{6,-2}{1.}{.6}{\small $\rho_{ss}$}{1};
            \draw[very thick,draw=red] (6,-3)--(6,-4);
            \draw[very thick,draw=red] (6,-1)--(6,0);
            \draw[very thick,draw=red] (5,0)--(6,0);
            \draw[very thick,draw=red] (5,-4)--(6,-4);
        \end{tikzpicture}\quad 
    \end{align}    
\end{result}
Hence, any (transient) quadratic contribution in $F[\psi_V]$ can only occur as a result of the effective photon generator $H$. Conversely, one finds that if that is the case, then the noisy QFI will be asymptotically linear. 

Let $\rho_V = \sum_i \lambda_i \ketbra{\lambda_i}$, and thus $\rho_W = \sum_i \lambda_i U_\theta^\dagger \ketbra{\lambda_i} U_\theta \equiv \sum_i \lambda_i \ketbra{\tilde{\lambda}_i}$. We have, 
\begin{equation}
    \dot{\rho}_V = \iota[H,\rho_V] + U_\theta \dot{\rho}_W U_\theta^\dagger
\end{equation}
Now, the mixed state QFI expression reads, 
\begin{equation}
    F[\rho] = \sum_{ij} \frac{2}{\lambda_i + \lambda_j} |\expect{\lambda_i}{\dot{\rho}}{\lambda_j}|^2 
\end{equation}
\begin{equation}
    F[\rho_V] = \underbrace{\sum_{ij}\frac{2(\lambda_i - \lambda_j)^2}{\lambda_i + \lambda_j} |\expect{\lambda_i}{H}{\lambda_j}|^2}_{F[\rho_W,H]}  + F[\rho_W] + 2\sum_{ij}\frac{\lambda_i-\lambda_j}{\lambda_i + \lambda_j}\text{Im}[\expect{\lambda_i}{H}{\lambda_j}\expect{\tilde{\lambda}_j}{\dot{\rho}_W}{\tilde{\lambda}_i}]
\end{equation}
Now, we have that $F[\rho_W,H] \leq 4\var[H]$ and $F[\rho_W] \leq 4T \alpha_W$. Notably, all terms other than $F[\rho_W,H]$ are strictly linear in $T$, without any quadratic transients. Now, we use the following result to lower bound $F[\rho_W,H]$.

\begin{theorem}[Theorem 3 of \cite{MarvianVariational}] \label{thm:PurifiedQFI}
    The QFI of the mixed state $\rho_\theta = U_\theta \rho U_\theta^\dagger$ with $U_\theta = e^{\iota\theta H_R}$ is equal to 
    \begin{equation}
        F[\rho_\theta] = \min_{H_{\text{tot}}} \var_{\Psi}[H_{\text{tot}}]
    \end{equation}
    where $\Psi \equiv \Psi_{SR}$ is any purification of $\rho$ with $\rho_R = \tr_R[\ket{\Psi}\bra{\Psi}_{SR}]$ and the minimization is over all $H_S$ such that $H_{\text{tot}} = H_S \otimes \mathbbm{1}_R + \mathbbm{1}_S \otimes H_R$.
\end{theorem}
 
Recalling the classic result bounding operator covariances in terms of mutual information.

\begin{theorem}[\cite{wolf2008area}]\label{thm:CovMI}
    For a bipartite state $\rho_{SR}$ we have 
    \begin{equation}
        \text{Cov}[O_S,O_R]^2 \leq \|O_S\|^2\|O_R\|^2 I[S:R]
    \end{equation}
    where $I[S:R]$ is the mutual information.
\end{theorem}

Combining these two, we have the following Lemma.

\begin{lem}\label{lem:HKSmixedqfibound}
    $F[\rho_W,H] \geq \var[H] - c f(T)$ with $f(T) < T$ for $T < \tau^*$ and $\lim_{T\to\infty}f(T) = \tau^*$ with $c = O((\log{D})^k)$ for some finite $k>0$ depending on the parameter encoding.
\end{lem}

\begin{proof}
    Consider $\ket{\psi_W}$ to be the purification of $\rho_W$. Now consider $H_{\text{tot}} = H_S + H_R$. The variance of this object is (all (co-)variances are on $\psi_W$),
    \begin{equation}
        \var[H_{\text{tot}}] = \var[H_S] + \var[H_R] + 2\text{cov}[H_S,H_R]
    \end{equation}
    Now, note that $\var[H_R] = \var_{\rho_W}[H_R]$ is the QFI term we are after. Additionally, $\var[H_S] \geq 0$. Thus, the only term which can decrease the total variance is the covariance term. However, note that $H_R = \sum_\tau h^R_\tau$ from Lemma~\ref{lem:HKSImplies}. Thus, we get, 
    \begin{equation}
        \text{cov}[H_S,H_R] = \sum \text{cov}[H_S,h_\tau^R] \geq - \left|\sum \text{cov}[H_S,h_\tau^R]\right| \geq -\sum \left|\text{cov}[H_S,h_\tau^R]\right| \geq - \|H_S\| \|h\| \sum I[S:\tau]
    \end{equation}
   where we used Theorem \ref{thm:CovMI} in the last step. If the correlation length is $\tau^*$, we have, 
    \begin{equation}
        \text{cov}[H_S,H_R] \geq -\|H_S\| \|h\| \sum e^{-\tau/\tau^*} \geq -\|H_S\| \|h\| \frac{e^{-1/\tau^*}}{1-e^{-1/\tau^*}} \left(1-e^{-T/\tau^*}\right)
    \end{equation}
    Thus, we have shown that for any purification, and any $H_S$ one has 
    \begin{equation}
        \var[H_{\text{tot}}] \gtrsim \var[H_R] - \|H_S\| \|h\| f(T)
    \end{equation}
    with 
    \begin{equation}
        f(T) = \frac{e^{-1/\tau^*}}{1-e^{-1/\tau^*}} \left(1-e^{-T/\tau^*}\right)
    \end{equation}
    The operator norm $\|H_S\|$ cannot scale arbitrarily with the same argument as of Lemma.~\ref{lem:mixedqfibound}. 
    By Theorem \ref{thm:PurifiedQFI}, we thus have that, 
    \begin{equation} 
        F[\rho_W,H] \geq \var[H] - c f(T)
    \end{equation}
\end{proof}
Crucially, note that we did not need to invoke mixing of the dynamics to the steady state to prove this, and this is valid in the transient regime of $0 < T < \tau^*$ as well. That is, the transient quadratic part of the $\var[H]$ can be obtained by measuring the photons alone if HKS is satisfied.

\end{document}